\documentclass[a4paper,UKenglish,cleveref, autoref, thm-restate]{lipics-v2021}

\pdfoutput=1 %
\hideLIPIcs  %

\nolinenumbers

\newif\ifFull
\Fullfalse

\usepackage[normalem]{ulem}
\usepackage{caption}
\usepackage{subcaption}
\usepackage{tikz}

\usepackage{pgfplots}
\usepackage{framed,enumitem} 
\usepackage[lineno]{algorithm}
\usepackage{numprint}
\npdecimalsign{.}
\npthousandsep{\,}
\usepackage[noend]{algpseudocode}
\usepackage{colortbl}
\usepackage{booktabs}
\usepackage{tabularray}
\usepackage{makecell}
\usepackage[symbol]{footmisc}

\DeclareMathOperator*{\argmax}{argmax}

\usepackage{amsbsy}
\usepackage{mathtools}

\newtheorem*{theorem*}{Theorem}
\newcommand{\Is}       {\coloneqq}
\newcommand{\set}[1]{\left\{ #1\right\}}

\newcommand{\gilt}{:}
\newcommand{\sodass}{\,:\,}
\newcommand{\setGilt}[2]{\left\{ #1\sodass #2\right\}}

\def\MdR{\ensuremath{\mathbb{R}}}

\title{Buffered Streaming Edge Partitioning} 

\titlerunning{Buffered Streaming Edge Partitioning} %

\author{Adil {Chhabra}}{Heidelberg University, Germany }{adil.chhabra@stud.uni-heidelberg.de}{https://orcid.org/0009-0009-5726-9389}{}
\author{Marcelo {Fonseca Faraj}}{Heidelberg University, Germany}{marcelofaraj@informatik.uni-heidelberg.de}{https://orcid.org/0000-0001-7100-236X}{}
\author{Christian {Schulz}}{Heidelberg University, Germany}{christian.schulz@informatik.uni-heidelberg.de}{https://orcid.org/0000-0002-2823-3506}{}
\author{Daniel {Seemaier}}{Karlsruhe Institute of Technology, Germany}{daniel.seemaier@kit.edu}{https://orcid.org/0000-0002-1997-1304}{}

\authorrunning{A. Chhabra et al.}

\Copyright{Adil {Chhabra}, Marcelo {Fonseca Faraj}, Daniel {Seemaier}, and Christian {Schulz}}



\category{}

\acknowledgements{We acknowledge support by DFG grant SCHU 2567/5-1. 
\flag[2cm]{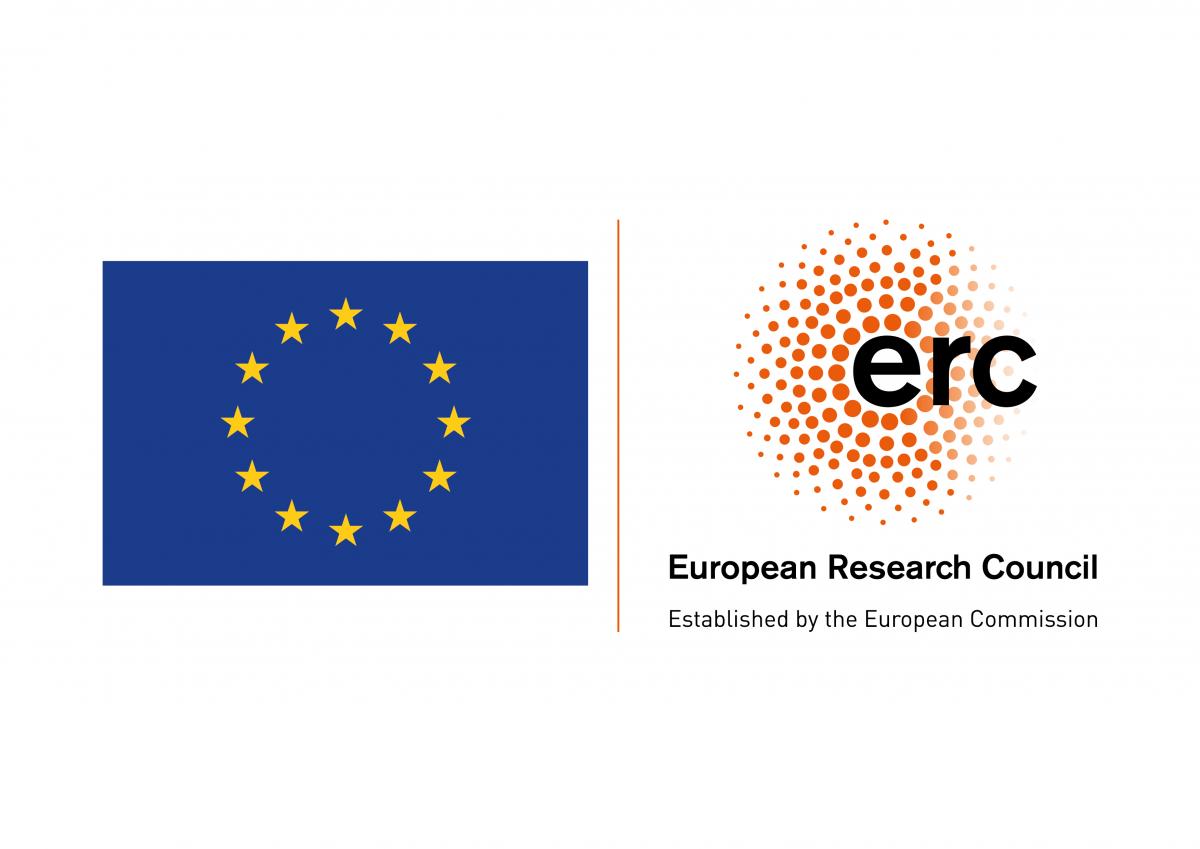}
This project has received funding from the European Research Council (ERC) under the European Union’s Horizon 2020 research and innovation programme (grant agreement No. 882500).
Moreover, we would like to acknowledge Dagstuhl Seminar 23331 on Recent Trends in Graph Decomposition.
}

\EventEditors{Leo Liberti}
\EventNoEds{1}
\EventLongTitle{21st International Symposium on Experimental Algorithms (SEA 2024)}
\EventShortTitle{SEA 2024}
\EventAcronym{SEA}
\EventYear{2023}
\EventDate{July 24--26, 2023}
\EventLocation{Vienna, Austria}
\EventLogo{}
\SeriesVolume{265}
\ArticleNo{15}

\begin{document}

\maketitle

\begin{abstract}
Addressing the challenges of processing massive graphs, which are prevalent in diverse fields such as social, biological, and technical networks, we introduce \texttt{HeiStreamE} and \texttt{FreightE}, two innovative (buffered) streaming algorithms designed for efficient edge partitioning of large-scale graphs. \texttt{HeiStreamE} utilizes an adapted Split-and-Connect graph model and a Fennel-based multilevel partitioning scheme, while \texttt{FreightE} partitions a hypergraph representation of the input graph. Besides ensuring superior solution quality, these approaches also overcome the limitations of existing algorithms by maintaining linear dependency on the graph size in both time and memory complexity  with no dependence on the number of blocks of partition. Our comprehensive experimental analysis demonstrates that \texttt{HeiStreamE} outperforms current streaming algorithms and the re-streaming algorithm \texttt{2PS} in partitioning quality (replication factor), and is more memory-efficient for real-world networks where the number of edges is far greater than the number of vertices. Further, \texttt{FreightE} is shown to produce fast and efficient partitions, particularly for higher numbers of partition blocks. 
\end{abstract}

\vfill \pagebreak

\setcounter{page}{1}%

\section{Introduction}
\label{sec:introduction}

Complex, large graphs, often composed of billions of entities, are employed across multiple fields to model social, biological, navigational, and technical networks. However, processing huge graphs requires extensive computational resources, necessitating the parallel computation of graphs on distributed systems. Large graphs are partitioned into sub-graphs distributed among $k$ processing elements (PEs). PEs perform computations on a portion of the graph, and communicate with each other through message-passing. \emph{Graph partitioning} models the distribution of graphs across PEs such that each PE receives approximately the same number of components (vertices) and communication between PEs (via edges between them) is minimized. Edge partitioning, which outperforms traditional vertex partitioning on real-world power-law graphs~\cite{gonzalez2012powergraph,jure2008,low2012}, partitions edges into $k$ blocks such that vertex replication is minimized, hence minimizing the communication needed to synchronize vertex copies. \emph{Graph vertex and edge partitioning} are NP-hard~\cite{bourse-2014,Garey1974} and there can be no approximation algorithm with a constant ratio factor for general graphs unless P = NP~\cite{BuiJ92}. Thus, heuristic algorithms are used in practice. Further, due to data proliferation, streaming algorithms are increasingly being used to partition huge graphs quickly with low computational resources~\cite{alistarh2015streaming,awadelkarim2020prioritized,freight_paper,HeiStream,StreamMultiSection,hoang2019cusp,jafari2021fast,mayer2018adwise,tacsyaran2021streaming,tsourakakis2014fennel}.

Streaming edge partitioning entails the sequential loading of edges for immediate
assignment to blocks. One-pass streaming edge partitioners permanently assign edges to blocks during a single sequential pass over the graph’s data stream~\cite{petroni2015hdrf,xie2014distributed}. Alternatively, buffered streaming algorithms receive and store a buffer of vertices along with their edges before making assignment decisions, thus providing information about future vertices~\cite{HeiStream,mayer2018adwise}, and re-streaming algorithms gather information about the global graph structure~\cite{2PS2020,psl}. With a few exceptions~\cite{psl,xie2014distributed} most streaming edge partitioners have a high time complexity due to a linear dependency on the number of blocks $k$. However, in recent years, high $k$ values are frequently used in graph partitioning due to the increasing size of graphs, complexity of computations, and availability of processors. An existing re-streaming edge partitioner, \texttt{2PS-L}~\cite{psl}, achieves a linear runtime independent of $k$, but produces lower solution quality than state-of-the-art partitioners and has a linear memory dependence on $k$. Thus, there remains potential to explore high-quality streaming edge partitioners without a runtime and memory dependency on $k$.

\textbf{Contribution.} We propose \texttt{HeiStreamE}, a buffered streaming algorithm for edge partitioning that leverages the performance efficacy of multilevel algorithms. By employing an adapted version of the SPlit-And-Connect (SPAC) model~\cite{SPAC2017} and solving it with a Fennel-based multilevel scheme~\cite{HeiStream}, our algorithm produces superior solution quality while maintaining time and memory complexities that are linearly dependent on the size of the graph and independent of $k$. Our results establish the superiority of \texttt{HeiStreamE} over all current streaming algorithms, and even the re-streaming algorithm \texttt{2PS}~\cite{2PS2020,psl}, in replication factor. These outcomes highlight the considerable potential of our algorithm, positioning it as a promising tool for edge partitioning.
We additionally provide an implementation of an efficient streaming edge partitioner, \texttt{FreightE}, which uses streaming hypergraph partitioning~\cite{freight_paper} to partition edges on the fly. Our experiments demonstrate that \texttt{FreightE} is significantly faster \hbox{than all competing algorithms, especially for high $k$ values.}

\vfill \pagebreak

\section{Preliminaries}
\label{sec:preliminaries}

\subsection{Basic Concepts}
\label{subsec:basic_concepts}

\subparagraph*{(Hyper)Graphs.}
Let $G=(V=\{0,\ldots, n-1\},E)$ be an \emph{undirected graph} with no multiple or self-edges, such that $n = |V|$, $m = |E|$.
Let $c: V \to \MdR_{\geq 0}$ be a vertex-weight function, and let $\omega: E \to \MdR_{>0}$ be an edge-weight function.
We generalize $c$ and $\omega$ functions to sets, such that $c(V') = \sum_{v\in V'}c(v)$ and $\omega(E') = \sum_{e\in E'}\omega(e)$.
An edge $e = (u, v)$ is said to be \emph{incident} on vertices $u$ and $v$. Let $N(v) = \setGilt{u}{(v,u) \in E}$ denote the neighbors of $v$. 
A graph $S=(V', E')$ is said to be a \emph{subgraph} of $G=(V, E)$ if $V' \subseteq V$ and $E' \subseteq E \cap (V' \times V')$. 
If $E' = E \cap (V' \times V')$, $S$ is an \emph{induced} subgraph.
Let $d(v)$ be the degree of vertex $v$ and $\Delta$ be the maximum degree of $G$. %
Let $H = (V, E)$ be an \emph{undirected hypergraph} with $n = |V|$ vertices, $m = |E|$ hyperedges or nets. A net, unlike an edge of a graph, may consist of more than two vertices, and is defined as a subset of $V$.

\subparagraph*{Partitioning.}
Given a number of \emph{blocks} $k \in \mathbb{N}_{\geq 1}$, and an undirected (hyper)graph with \emph{positive} edge weights, the \emph{(hyper)graph partitioning} problem pertains to the partitioning of a (hyper)graph into $k$ smaller (hyper)graphs by assigning the vertices (\emph{vertex partitioning}) or (hyper)edges (\emph{edge partitioning}) of the graph to $k$ mutually exclusive blocks, such that the blocks have roughly the same size and the particular objective function is minimized or maximized. More precisely, a \emph{\mbox{$k$-vertex partition}} of a (hyper)graph partitions $V$ into $k$ blocks $V_1, \dots, V_k$ such that $V_1 \cup \cdots \cup V_k = V$ and $V_i \cap V_j = \emptyset$ for $i \neq j$. 
The \emph{edge-cut} (resp. \emph{cut net}) of a $k$-partition consists of the total weight of the \emph{cut edges} (resp. \emph{cut nets}), i.e., (hyper)edges crossing blocks.
More formally, let the edge-cut (resp. \emph{cut net}) be $\sum_{i<j}\omega(E')$, in which $E' \coloneqq $ $\big\{e\in E, \exists (u,v) \subseteq e : u\in V_i,v\in V_j, i \neq j\big\}$ is the~\emph{cut-set}, i.e.,~the set of all \emph{cut edges} (resp. \emph{cut nets}).
The \emph{balancing constraint} demands that the sum of vertex weights in each block does not exceed a threshold associated with some allowed \emph{imbalance}~$\varepsilon$.
More specifically, \hbox{$\forall i \in \{1,\dots,k\} \gilt$ $c(V_i)\leq L_{\max}\Is \big\lceil(1+\varepsilon) \frac{c(V)}{k} \big\rceil$}. 
For each net $e$ of a hypergraph, $\Lambda(e) \coloneqq \{V_i \mid V_i \cap e \neq \emptyset\}$ denotes the \emph{connectivity set} of $e$.
Further, the \emph{connectivity} $\lambda(e)$ of a net~$e$ is the cardinality of its connectivity set, i.e., $\lambda(e) \coloneqq |\Lambda(e)|$.
The so-called \emph{connectivity} metric ($\lambda-1$) is computed as $\sum_{e\in E'} (\lambda(e) -1)~\omega(e)$, where $E'$ is the cut-set.

Similarly, a \emph{\mbox{$k$-edge partition}} of a graph partitions the edge set $E$ into $k$ blocks \hbox{$E_1$,\ldots,$E_k$} \hbox{such that $E_1\cup\cdots\cup E_k=E$} and $E_i\cap E_j=\emptyset$ for $i\neq j$. In edge partitioning, a common objective function is the minimization of the \emph{replication factor}, which is defined as the number of replicated vertices divided by the total number of vertices in the graph. 
Formally, we define the set $V(E_i) = \big\{ v \in V \mid \exists u \in V: (u, v) \in E_i \big\}$ for each partition $E_i$ as the number of vertices in $V$ that have at least one edge incident on them that was assigned to block $E_i$. Taking the sum of $|V(E_i)|$ over all $k$ gives us the total number of vertex replicas generated by the partition. \emph{Replication factor} is then defined as
$RF(E_1, E_2, \dots, E_k) = \frac{1}{n}\sum_{i = 1,\dots,k}|V(E_i)|$.
Intuitively, a minimized \emph{replication factor} suggests that vertices are replicated in minimum blocks. Minimum vertex replication, in turn, results in lower synchronization overhead in distributed graph processing due to reduced exchange of vertex state across blocks. 

\begin{figure}[t]
	\centering
	\includegraphics[width=.5\textwidth]{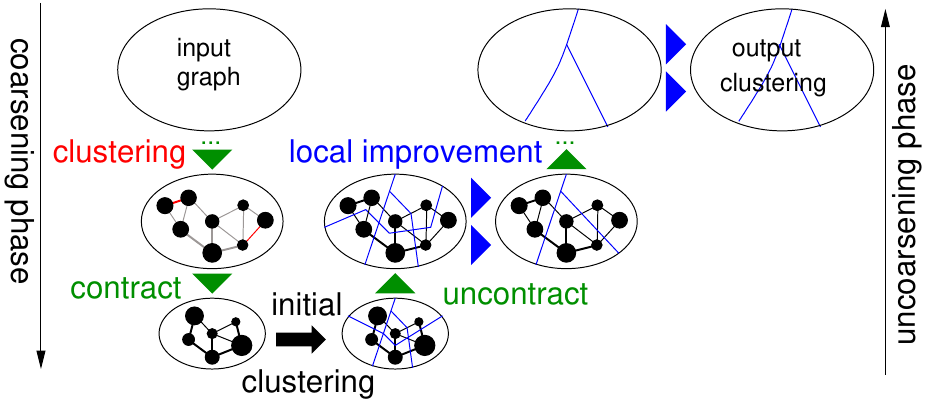}
	\caption{Multilevel scheme.}
	\label{fig:MSGC}
\end{figure}

\vspace*{-0.25cm}
\subparagraph*{Multilevel Scheme.}
\label{subsec:Multilevel Scheme}
A successful heuristic for \emph{vertex partitioning} is the \emph{multilevel}~\cite{more_recent_advances_hgp} approach, illustrated in Figure~\ref{fig:MSGC}.
It recursively computes a clustering and \emph{contracts} it to \emph{coarsen} the graph into smaller graphs that maintain the same basic structure as the input graph.
An \emph{initial partitioning} algorithm is applied to the smallest (\emph{coarsest}) graph and then the contraction is undone.
At each level, a \emph{local search} method is used to improve the partitioning induced by the coarser level. 
\emph{Contracting} a cluster of vertices $C=\set{u_1, \ldots, u_{\ell}}$ involves replacing them with a new vertex~$v$ whose weight is the sum of the weights of the clustered vertices and is connected to all elements  $w \in \bigcup_{i=1}^{\ell} N(u_i)$, \hbox{$\omega(\set{v,w})=\sum_{i=1}^{\ell}\omega(\set{u_i,w})$}.
This ensures that the transfer of partitions from a coarser to a finer level maintains the edge-cut.
The \emph{uncontraction} of a vertex undoes the contraction. 
Local search moves vertices between blocks to reduce the objective.  

\subparagraph*{SPAC Transformation.} The SPAC transformation~\cite{SPAC2017} provides a means to employ a vertex partitioning tool on a transformed graph $G'$, which is derived from the original graph $G$, and subsequently apply the derived vertex partition to establish an edge partition for $G$.
The transformation assumes an undirected, unweighted graph $G=(V,E)$ as input. The SPAC graph $G' = (V',E')$ is then constructed in two phases: 
In the \emph{split} phase, each vertex $v \in V$ generates $d(v)$ split vertices \hbox{$S_v \coloneqq \{v'_1,...,v'_{d(v)}\} \subseteq V'$}.
The \emph{connect} phase introduces two kinds of edges in $E'$, \emph{dominant} edges and \emph{auxiliary} edges.
First, it assigns a \emph{dominant} edge $e' = (\mu'_i, v'_j)$ in $G'$ for each edge $e = (u, v)$ in $G$.
Dominant edges are created with infinite weight $\omega(e') = \infty$.
Second, it introduces as many \emph{auxiliary} edges $e''$ as necessary to create a path connecting the vertices in the set $S_v$ for each vertex $v \in V$.
Auxiliary edges are created with unitary weight $\omega(e'') = 1$.
A visual representation of the SPAC transformation is provided in Figure \ref{fig:cspac}.
Due to the infinite weight of dominant edges, vertex partitioning tools usually refrain from splitting them, causing both endpoints of a dominant edge to be grouped in the same block (alternatively, straightforward heuristics can compel both endpoints of dominant edges to belong to the same block).
Next, the block assigned to both endpoints of each dominant edge is assigned to the edge in $G$ that induced the corresponding dominant edge, thereby resulting in an edge partition of $G$. 
The SPAC method is particularly effective in practical scenarios and yields a sound, provable approximation factor under specific balance constraints.
Specifically, it approximates the balanced edge partitioning problem within $\mathcal{O}(\Delta \sqrt{\log{n} \log{k}})$, where $\Delta$ is the maximum degree of $G$~\cite{SPAC2017}.

\subparagraph*{Buffered Streaming.}
\label{subsec:Computational Model}
In the buffered streaming model, which is an extended version of the one-pass model, we load a $\delta$-sized buffer or batch of input vertices along with their edges. We make block assignment decisions only after the entire batch has been loaded. In practice, the parameter $\delta$ can be chosen in accordance with memory available on the machine. In our contribution, we use a fixed $\delta$ throughout the algorithm. For a predefined batch size of $\delta$, we load and repeatedly partition $\left \lceil{n/\delta}\right \rceil$ batches.

\subsection{Related Work}
\label{subsec:related_work}
We refer the reader to recent surveys on (hyper)graph partitioning for relevant literature~\cite{SPPGPOverviewPaper,more_recent_advances_hgp,DBLP:reference/bdt/0003S19}. 
 Here, we focus on the research on streaming vertex and edge partitioning.
Most high-quality vertex partitioners for real-world graphs use a multilevel scheme, including \texttt{KaHIP}~\cite{kabapeE}, \texttt{METIS}~\cite{karypis1998fast}, \texttt{Scotch}~\cite{Pellegrini96experimentalanalysis}, and \texttt{(mt)-KaHyPar}~\cite{mt-kahypar-d,schlag2016k}. Edge partitioning has been solved directly with multilevel hypergraph partitioners, including \texttt{PaToH}~\cite{ccatalyurek2011patoh}, \texttt{hMETIS}~\cite{hMetis}, \texttt{KaHyPar}~\cite{dSPAC2019}, \texttt{Mondriaan}~\cite{brendan2002mondriaan}, \texttt{MLPart}~\cite{alpert1998}, \texttt{Zoltan}~\cite{devine2006}, \texttt{SHP}~\cite{igor2017}, \texttt{UMPa}~\cite{umpadimacs}, and \texttt{kPaToH}~\cite{ccatalyurek2011patoh}.

\subparagraph*{Streaming (Hyper)Graph Vertex Partitioning.}
Tsourakakis~et~al.~\cite{tsourakakis2014fennel} introduce \texttt{Fennel}, a one-pass partitioning heuristic adapted from the clustering objective \emph{modularity}~\cite{brandes2007modularity}. \texttt{Fennel} minimizes edge-cuts by placing vertices in partitions with more neighboring vertices. \texttt{Fennel} assigns a vertex $v$ to the block $V_i$ that maximizes the Fennel gain function $|V_i\cap N(v)|-f(|V_i|)$, where $f(|V_i|)$ is a penalty function to respect a balancing threshold. The authors define the Fennel objective with $f(|V_i|) = \alpha \gamma \cdot |V_i|^{\gamma-1}$, in which~$\gamma$ is a free parameter and $\alpha = m \frac{k^{\gamma-1}}{n^{\gamma}}$. After parameter tuning, the authors define $\gamma=\frac{3}{2}$ and $\alpha=\sqrt{k}\frac{m}{n^{3/2}}$. The time complexity of the algorithm depends on $k$ and is given by $\mathcal{O}(nk+m)$. 
Stanton and Kliot~\cite{stanton2012streaming} propose \texttt{LDG}, a greedy heuristic for streaming vertex partitioning. \texttt{ReLDG} and \texttt{ReFennel} are re-streaming versions of \texttt{LDG} and \texttt{Fennel}~\cite{nishimura2013restreaming}. Prioritized re-streaming optimizes the ordering of the streaming process~\cite{awadelkarim2020prioritized}. 
Faraj and Schulz~\cite{HeiStream} propose \texttt{HeiStream}, which uses a generalized weighted version of the Fennel gain function in a buffered streaming approach. 
Eyubov et al.~\cite{freight_paper} introduce \texttt{FREIGHT}, a streaming hypergraph partitioner that adapts the Fennel objective function to partition vertices of a hypergraph on the fly.

\subparagraph*{Streaming Edge Partitioning.}
One-pass streaming edge partitioners include hashing-based partitioners like \texttt{DBH}~\cite{xie2014distributed}, constrained partitioners like \texttt{Grid} and \texttt{PDS}~\cite{jain2013gridpds}, and \texttt{HDRF}, proposed by Petroni et al.~\cite{petroni2015hdrf}. \texttt{HDRF} exploits the skewed degree distribution of real-world graphs by prioritizing vertex replicas of high-degree vertices. \texttt{HDRF} outperforms \texttt{DBH}, \texttt{Grid} and \texttt{PDS} in solution quality, but has a longer runtime. Zhang et al.~\cite{zhang2017graph} introduced \texttt{SNE}, a streaming version of the in-memory edge partitioner \texttt{NE} that utilizes sampling methods. \texttt{SNE} produces better solution quality than \texttt{HDRF}, but with increased memory consumption and runtime~\cite{psl}. In contrast to one-pass streaming models, \texttt{RBSEP} uses a buffered approach to postpone assignment decisions for edges with limited neighborhood partitioning information available during streaming~\cite{RBSEP}. Additionally, Mayer et al.~\cite{mayer2018adwise} introduce \texttt{ADWISE}, a window-based streaming edge partitioner, which uses a dynamic window size that adapts to runtime constraints. %

Mayer et al.~\cite{2PS2020} subsequently propose \texttt{2PS-HDRF}, a two-phase re-streaming algorithm for edge partitioning, using \texttt{HDRF} as the scoring function in its final partitioning step. The first phase uses a streaming clustering algorithm to gather information about the global graph structure; in the second phase, the graph is re-streamed and partitioned, using information obtained from clustering to make edge partitioning decisions. 
Mayer et al.~\cite{psl} modify \texttt{2PS-HDRF} to propose \texttt{2PS-L}, which runs in time independent of $k$. \texttt{2PS-L} switches from \texttt{HDRF} to a new scoring function in the final partitioning step to remove its dependency on $k$, and thus achieves a time complexity of $\mathcal{O}(|E|)$. 
\texttt{2PS-L} outperforms \texttt{ADWISE}; it is faster than \texttt{HDRF} and \texttt{2PS-HDRF}, particularly at large $k$ values, but has lower solution quality. \texttt{2PS-HDRF} achieves 50\% better solution quality than \texttt{2PS-L}~\cite{psl}.

Sajjad~et~al.~\cite{sajjad2016boosting} propose \texttt{HoVerCut}, a platform for streaming edge partitioners, which can scale in multi-threaded and distributed systems by decoupling the state from the partitioner.
Hoang~et~al.~\cite{hoang2019cusp} propose \texttt{CuSP}, a distributed and parallel streaming framework to partition edges based on user-defined policies.
\texttt{CuSP} is programmable and can express common streaming edge partitioning strategies from the literature.

\section{Buffered Streaming Edge Partitioning}
\label{sec:Buffered Streaming Edge Partitioning}

In this section, we present our algorithms, \texttt{HeiStreamE} and \texttt{FreightE}.
First, we provide an overview of \texttt{HeiStreamE}'s iterative structure. Subsequently, we detail its input format and buffered graph model, and describe how it uses multilevel vertex partitioning to solve this model. Lastly, we discuss how \texttt{FreightE} builds a hypergraph representation to partition edges using a streaming hypergraph partitioner.
\begin{algorithm}[b]
	\caption{Overall Structure of \texttt{HeiStreamE}}
	\label{alg:overall}
	\begin{algorithmic}[1] %
		\For{$b \in \{1,\ldots,\lceil\frac{n}{\delta}\rceil\}$}
			\State Load subgraph $G_b$ from input graph $G$
			\State Build model $\beta_b$ from $G_b$
			\State Run multilevel vertex partitioning on $\beta_b$ 
			\State Permanently assign corresponding edges of $G$ in $G_b$ to blocks
		\EndFor
	\end{algorithmic}
\end{algorithm}

\subsection{Overall Algorithm}
\label{subsec:Overall Algorithm}

Our framework draws inspiration from \texttt{HeiStream}~\cite{HeiStream}.
We slide through the input graph $G$ by iteratively performing the following series of operations until all the edges of $G$ are assigned to blocks.
First, we load a batch composed of $\delta$ vertices and their associated neighborhood, thereby obtaining a subgraph~$G_b$ contained within the graph~$G$.
This operation yields edges connecting vertices within the current batch, and edges connecting vertices in the current batch to vertices streamed in previous batches.
Second, we build a model $\beta_b$ corresponding to $G_b$, where the edges of $G_b$ are transformed into vertices. Additionally, we incorporate a representation of block assignments from previous batches into $\beta_b$.
Third, we partition $\beta_b$ using a multilevel vertex partitioning algorithm that has been shown to be effective in the context of buffered streaming~\cite{HeiStream}.
We conclude by permanently assigning the edges in $G$ that correspond to vertices in our model $\beta_b$ to \hbox{their respective blocks}.
Algorithm \ref{alg:overall} summarizes the general structure of \texttt{HeiStreamE}, which is illustrated in Figure~\ref{fig:structure}.

\begin{figure}[t]
	\centering
	\includegraphics[width=0.8\textwidth]{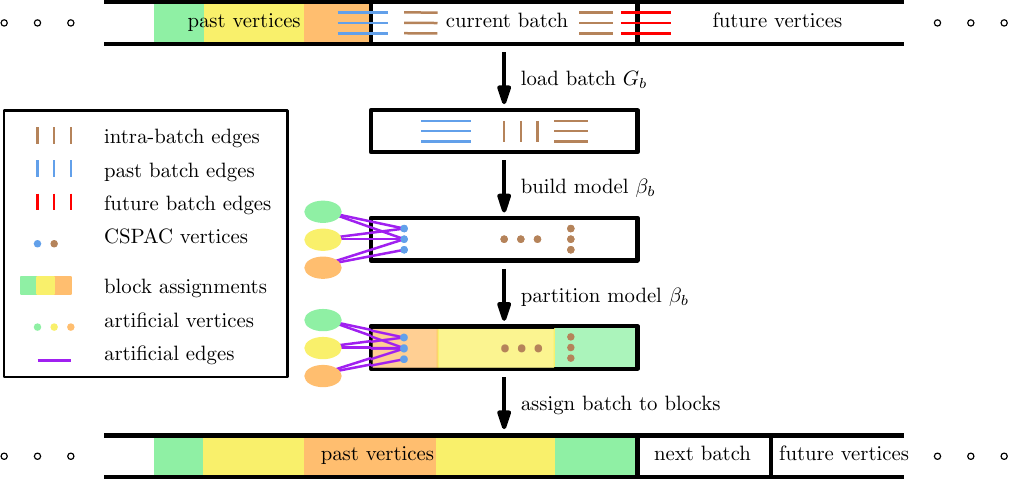}
	\caption{
	Detailed structure of \texttt{HeiStreamE}. 
	The algorithm starts by loading a batch graph~$G_b$ consisting of vertices and their edges to the current batch and previous batches.
	Subsequently, it builds a meaningful model~$\beta_b$ from the batch graph, transforming edges into vertices, and incorporating a synthetic representation of the assignments made in previous batches.
	This model is then partitioned using a multilevel algorithm.
	Lastly, the edges from the loaded batch, which correspond to vertices in the partitioned batch model, are permanently assigned to blocks.
	This process is repeated for subsequent batches until the entire graph has been partitioned.
	}
	\label{fig:structure}
\end{figure}

\subsection{Input and Batch Format}
\label{subsec:Input and Batch Format}

\texttt{HeiStreamE} uses a vertex-centric input format and a buffered streaming approach,  which refers to the sequential process of loading and handling the input graph in batches.
Within each batch, it loads $\delta$ vertices one at a time along with their neighborhood, where $\delta$ is a parameter that defines the buffer size. Similar input formats are commonly used in streaming algorithms for vertex partitioning~\cite{awadelkarim2020prioritized,freight_paper,HeiStream,StreamMultiSection,jafari2021fast,tsourakakis2014fennel}
and are consistent with graph formats commonly found in publicly available real-world graph datasets, such as the METIS format.
Each batch $b \in \{1,\ldots,\lceil\frac{n}{\delta}\rceil\}$ corresponds to a subgraph denoted as $G_b=(V_b,E_b)$ within the graph $G$.
This subgraph is constructed as follows. 
Its vertex set $V_b$ includes the $\delta$ vertices from the current batch, labeled within the domain $[0,\delta -1]$, as well as the $p$ vertices from past batches that share at least one edge with vertices in the current batch, labeled within the domain $[\delta,\delta + p - 1]$.
Similarly, its edge set $E_b$ comprises of edges with one endpoint in the current batch and the other endpoint in either the current batch or previous batches, expressed as $E_b = \{(u,v) \in E \mid \text{Batch}(u)=b \land \text{Batch}(v)\leq b\}$.
Edges with an endpoint in future batches are discarded, ensuring that each edge belongs exclusively to a unique batch graph $G_b$ across all batches.

\subsection{Model Construction}
\label{subsec:Model Construction}

Our model construction consists of two steps, detailed in this section.
First, using the batch graph $G_b$, we create a corresponding \emph{Contracted SPlit-And-Connect} (CSPAC) graph, denoted as $S^*_b$. 
In $S^*_b$, edges from $G_b$ are directly represented as vertices, while vertices from $G_b$ are indirectly represented as edges. 
Subsequently, we create a graph model $\beta_b$ based on $S^*_b$, which incorporates a representation of block assignments from prior batches.

\subparagraph*{CSPAC Transformation.}
\label{par:CSPAC Transformation}

The CSPAC transformation is a faster and more condensed variant of the SPAC transformation conceived by Li~et~al.~\cite{SPAC2017}, as described in Section \ref{subsec:basic_concepts}. Below we explain how the SPAC transformation evolves into the CSPAC transformation, and how this transformation is applied to~$G_b$ \hbox{to yield the CSPAC graph $S^*_b$}.

The SPAC graph $G'$ of a graph $G$ consists of dominant edges, i.e., edges with a weight of infinity that have a one-to-one correspondence with edges of the original graph, and auxiliary edges, which define a path between vertices of the SPAC graph. The CSPAC transformation is derived from the SPAC graph by contracting the dominant edges into vertices.
Due to construction, the dominant edges in $G'$ do not share any endpoints; they effectively form a matching, ensuring a consistent contraction. Further, each endpoint of the auxiliary edges is incident to a single dominant edge.
Thus, the contraction of all dominant edges produces a coarser graph, in which every vertex represents a unique edge in the original graph $G$, and their connections correspond to the auxiliary edges in graph $G'$.
The CSPAC transformation is illustrated in Figure \ref{fig:cspac}.

\begin{figure}[t]
	\centering
	\includegraphics[width=0.8\textwidth]{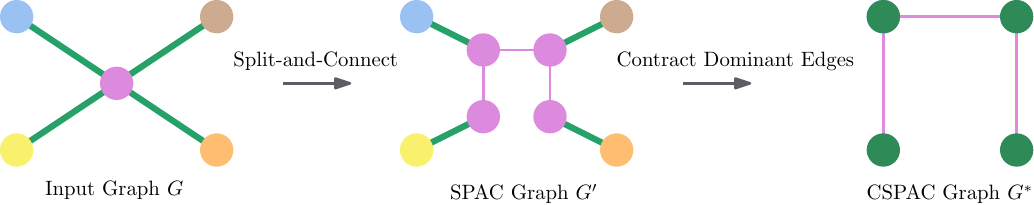}
	\caption{
		Building SPAC and CSPAC Graphs:
		The SPAC graph, denoted as $G'$, features $d(v)$ split vertices for every vertex $v$ in the original graph $G$.
		These split vertices are represented in the same color as the vertex they originate from.
		Every (thick green) edge from $G$ is directly converted into a distinct dominant (thick green) edge in $G'$ that connects corresponding split vertices.
		The auxiliary (thin) edges in $G'$, which create a path between split vertices, are depicted in the same color as the split vertices they link.
		The CSPAC graph $G^*$ is formed by contracting the dominant edges of $G'$.
		The vertices in $G^*$ represent the edges in $G$, while the edges in $G^*$ mirror the auxiliary edges in $G'$.
	}
	\label{fig:cspac}
\end{figure}

In \texttt{HeiStreamE}, we derive the CSPAC graph $S^*_b$ directly from the batch graph $G_b$, bypassing the construction of the intermediary SPAC graph.
The procedure for this direct transformation is as follows. 
For each vertex $u \in V_b$, each of its outgoing edges $e = (u, v) \in E_b$ induces a vertex $u^*$ in $S^*_b$, if $u < v$ to avoid redundancy. As each vertex is constructed,  it is connected to at most two other vertices in $S^*_b$ that are induced by other outgoing edges of $u$ in $G_b$ to form a path.  
This way, each undirected edge  $e \in E_b$, which would have induced a dominant edge in the SPAC graph, is represented as a vertex in $S^*_b$. Further, each auxiliary edge of the SPAC graph is directly integrated into $S^*_b$ to form paths between vertices in $S^*_b$.
The direct construction of the CSPAC graph maintains the same computational complexity as the SPAC construction alone, specifically $\mathcal{O}(|V_b|+|E_b|)$, where $|V_b|$ and $|E_b|$ are the number of vertices and edges of $G_b$ respectively.
In practice, building $S^*_b$ directly from $G_b$ offers a time-saving advantage compared to the alternative method of building the SPAC graph and then \hbox{contracting the dominant edges}.

In each batch, $S^*_b$ has $|E_b|$ vertices and $2|E_b|-|V_b|$ edges, i.e., it is linear in the size of the batch graph $G_b$. In contrast to the SPAC graph, the CSPAC graph has fewer vertices and edges. Further, where a vertex partitioner might cut a dominant edge of the SPAC graph, every vertex partition of the CSPAC graph corresponds to a valid edge partition of the original graph. 
Thus, there is no need for a verification step to transform a vertex partition of $S^*_b$ into an edge partition of $G_b$.

Theorem \ref{theorem:epvp} shows that when computing a vertex partition of $S^*_b$ to minimize the edge-cut, the corresponding edge partition of $G_b$ will also have a minimized number of vertex replicas.
The SPAC approximation factor shown by Li~et~al.~\cite{SPAC2017} is also directly valid for $S^*_b$.

\begin{theorem}
	\label{theorem:epvp}
	For any vertex partition $vp(S^*_b)$ of the CSPAC graph $S^*_b$ with edge-cut $cost(vp(S^*_b))$, there exists a corresponding edge partition $ep(G_b)$ of the batch graph $G_b$ with a number of vertex replicas $cost(ep(G_b))$, satisfying $cost(ep(G_b)) \leq cost(vp(S^*_b))$, which establishes a lower bound spanning the set of possible CSPAC graphs $S^*_b$ \hbox{associated with $G_b$}.
\end{theorem}
\begin{proof}
	
	Our proof can be delineated through three consecutive claims.
	(i) The existence of a singular edge partition $ep(G_b)$ that corresponds to a specified vertex partition $vp(S^*_b)$.
	(ii) The general validity of the inequality $cost(ep(G_b)) \leq cost(vp(S^*_b))$.
	(iii) The validity of the equality $cost(ep(G_b)) = cost(vp(S^*_b))$ for a specific \hbox{CSPAC graph $S^*_b$}.
	
	Claim (i) trivially holds, as $ep(G_b)$ stems directly from $vp(S^*_b)$ by virtue of the one-to-one correspondence between edges in $G_b$ and vertices in $S^*_b$.
	For the proof of assertions (ii) and (iii), consider the following.
	(a) There can be no replicas of vertices with degree lower than 2.
	(b) Vertices in $G_b$ with a degree lower than 2 are not represented by any edges in $S^*_b$.
	(c) Vertices $u$ in $G_b$ that possess a degree $d(u) \geq 2$ are represented in $S^*_b$ through unique, edge-disjoint paths, each connecting the $d(u)$ vertices in $S^*_b$ corresponding to the \hbox{edges incident to $u$ in $G_b$}.
	
	To prove (ii), we show that the number of replicas of any vertex $u$ in $ep(G_b)$ does not exceed the number of cut edges directly induced by $u$ in $vp(S^*_b)$.
	From (a) and (b), it trivially holds for vertices with a degree lower than 2.
	Assuming there are $x > 0$ replicas of $u$ in $ep(G_b)$, it implies that edges incident on $u$ are distributed across $x+1$ nonempty blocks.
	According to (c), $u$ is uniquely represented by an edge-disjoint path connecting the vertices in $S^*_b$ that correspond to the edges of $u$ in $G_b$.
	When the vertices of a connected (sub)graph are partitioned into $x+1$ blocks, these blocks themselves are interconnected.
	Therefore, the edge-cut directly attributable to vertex $u$ is at least $x$, which \hbox{completes the proof for (ii)}.
	
	To prove (iii), we show how to build a valid CSPAC graph $S^*_b$, such that $cost(ep(G_b)) = cost(vp(S^*_b))$.
	For vertices of $G_b$ with a degree lower than 2, the equality trivially holds based on (a) and (b).
	For a vertex $u$ in $ep(G_b)$ with $x > 0$ replicas (edges distributed across $x+1$ nonempty blocks), we create $x+1$ independent paths connecting the vertices in $S^*_b$ that represent the edges incident to $u$ in $G_b$. Vertices in each path are then assigned to a common block.
	As these paths are between vertices of the same block, they have no cut edges.
	Subsequently, these paths are interlinked to form a unified path.
	This introduces exactly $x$ cut edges directly associated with vertex $u$, thereby concluding the proof for condition (iii) and \hbox{the overall proof}.
\end{proof}

\begin{figure}[t]
	\centering
	\includegraphics[width=0.8\textwidth]{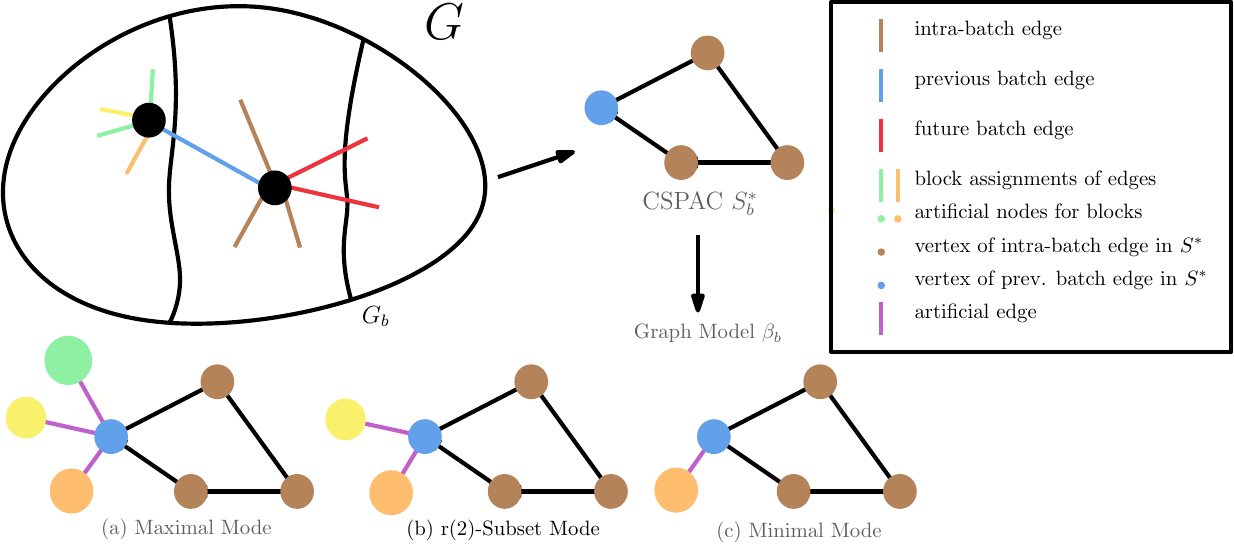}
	\caption{Graph model $\beta_b$ construction. $\beta_b$ is obtained by appending past assignment decisions to $S^*_b$. If a vertex of the current batch graph $u \in G_b$ has an edge $e = (u, v)$ to a previous batch (colored blue), we connect the CSPAC vertex $u^*$ induced by $e$ to artificial vertices representing blocks assigned to edges incident on $v$ as follows: (a) \emph{Maximal} Mode: $u^*$ connects to all blocks incident on~$v$ (b) \emph{$r$-Subset} Mode: $u^*$ connects to $r$ random blocks incident on~$v$ (c) \hbox{\emph{Minimal} Mode: $u^*$} connects to the block assigned to the most recently partitioned edge incident on $v$.}
	\label{fig:modes}
        \vspace*{-.5cm}
\end{figure}
\subparagraph*{Integrating Connectivity Information.}
\label{par:Integrating Connectivity Information}
Directly partitioning the CSPAC graph $S^*_b$ limits the partitioner's view to the current batch, as $S^*_b$ does not take into account block assignments from previous batches.  Specifically, when assigning a block to a vertex $u^* \in S^*_b$ induced by an edge $e = (u, v) \in E_b$, where $u$ is a vertex of the current batch and $v$ is a vertex of some previous batch, the partitioner might replicate $v$ into a new block in the absence of global information. 
To solve this, we extend the CSPAC graph $S^*_b$ with connectivity information derived from previous batch assignments to obtain the graph model $\beta_b$. 
We construct $\beta_b$ by augmenting $S^*_b$ with $k$ artificial vertices representing the $k$ partition blocks.

Each artificial vertex $i$ corresponds to an existing partition block $E_i$ in its current state, i.e., the weight of each artificial vertex $i$ is the weight of its corresponding block $E_i$ filled with edges that have been assigned to it \hbox{from previous batches}.
An edge $e = (u,v) \in E_b$ of the current batch is represented in our model $\beta_b$ as a vertex that connects to an artificial vertex $i$ if $v$ is a vertex of a previous batch whose some incident edge has already been assigned to block $E_i$.
However, in a streaming setting, there is limited knowledge of edge connectivity, i.e., it is not possible to directly determine which edges from previous batches are adjacent to edges from the current batch. 
To overcome this constraint, we maintain an array $B$ of size $n$ throughout the streaming process.
This array records, for each vertex $u \in V$, block(s) assigned to edges incident on it. Then, in $\beta_b$, a vertex $u^* \in S^*_b$ induced by an edge $e = (u,v) \in E_b$ is connected to any artificial vertices $i$ representing blocks recorded in $B[v]$. An added benefit of this model construction is that we do not need to maintain a data structure of size $m$ in our algorithm.

We propose three configurations for $\beta_b$, which vary in how they use $B$, that we name, in decreasing level of exactness, \emph{maximal}, \emph{$r$-Subset}, and \emph{minimal}. Each configuration has a runtime, memory and solution-quality trade-off. Let $u^*$ be a vertex in $S^*_b$ induced by an edge $e = (u,v) \in E_b$.
 In the \emph{maximal} model $\beta_b$, $u^*$ is connected to \emph{all} \hbox{artificial vertices $i$} representing blocks recorded in $B[v]$. The memory required to store the array $B$ in \hbox{the \emph{maximal} setup is $\mathcal{O}(nk)$}.
Theorem \ref{theorem:epvp_general} demonstrates that computing a vertex partition of the \emph{maximal} model $\beta_b$ to minimize edge-cut corresponds to an edge partition of $G_b$ that contributes a minimized number of new vertex replicas to the overall edge partition of the input graph $G$.
In the \emph{$r$-Subset} model,  $u^*$ is connected to \emph{a sample of r} artificial vertices $i$ representing blocks recorded in $B[v]$, where $r$ is a parameter.
Here, $B$ is identical as in the \emph{maximal} setup, but the model $\beta_b$ is more concise, allowing for a faster partitioning phase.
In the \emph{minimal} model, $u^*$ is connected to a single artificial vertex $i$ representing \emph{the most recent block} assigned to $v$ in a previous batch. We only store the latest assignment per vertex in $B$, thus the memory requirement is $\mathcal{O}(n)$, and the model $\beta_b$ is also more concise than in the other two setups. We illustrate the various \hbox{configurations in Figure~\ref{fig:modes}}.

\begin{theorem}
	\label{theorem:epvp_general}
	For any vertex partition $vp(\beta_b)$ of the \emph{maximal} model $\beta_b$ with edge-cut $cost(vp(\beta_b))$, there exists a corresponding edge partition $ep(G_b)$ of the batch graph $G_b$ that, when incorporated into the already partitioned section of the input graph $G$, introduces a number $\gamma$ of new vertex replicas, \hbox{satisfying $\gamma \leq cost(vp(\beta_b))$}.
\end{theorem}
\begin{proof}
	The only disparity between the CSPAC graph $S^*_b$ and the maximal model~$\beta_b$ is the presence of artificial vertices and edges in $\beta_b$.
	If there are no artificial edges in $\beta_b$, it implies that none of the edges in the batch graph $G_b$ are adjacent to edges in $G$ that have previously been assigned to blocks.
	In this scenario, Theorem \ref{theorem:epvp} provides sufficient grounds to establish the claim.
	If $\beta_b$ contains artificial edges, each artificial edge signifies a unique adjacency in $G$ of an edge in the batch graph $G_b$ to edges already assigned to blocks in previous batches.
	In this scenario, we complete the proof of the claim by noting that the number of cut artificial edges in $vp(\beta_b)$ cannot be less than the number of new replicas introduced exclusively for vertices \hbox{contained in previous batches}.
\end{proof}

\subsection{Partitioning}
\label{subsec:Partitioning}

In this section, we describe how we partition our model $\beta_b$. We employ a vertex partitioning algorithm on $\beta_b$, specifically an adapted version of the multilevel weighted Fennel algorithm utilized in \texttt{HeiStream}~\cite{HeiStream}.
We describe this algorithm and then present a modification to the initial partitioning phase to enhance our runtime performance.
Lastly, we discuss possible adaptations to the Fennel parameter $\alpha$.

\subparagraph*{Multilevel Fennel.}
\label{par:Multilevel Fennel}
Each per-batch graph model $\beta_b$ is partitioned using a multilevel partitioning scheme consisting of three successive phases, coarsening, initial partitioning, and uncoarsening, as depicted in Figure~\ref{fig:MSGC}. 

In the coarsening phase, the algorithm computes a clustering and contracts the graph at each level until it is smaller than a specified threshold.
These clusters are computed with label propagation while adhering to size constraints~\cite{pcomplexnetworksviacluster}.
The clustering algorithm ignores artificial vertices and edges during the coarsening phase, to ensure that they are never contracted and that previous block assignment decisions are available at the coarsest level.
For a graph with $n$ vertices and $m$ edges, a single round of size-constrained label propagation can be executed in $\mathcal{O}(n+m)$ time.
Initially, each vertex is placed in its own cluster, and in subsequent rounds vertices move to the cluster with the strongest connection, with a maximum of $L$ rounds, where $L$ is a tuning parameter.
The coarsening phase ends as soon as the graph has fewer vertices than a threshold of $\mathcal{O}(\max({|\beta_b|/k,k}))$, where $|\beta_b|$ is the number of vertices of $\beta_b$. 
For large buffer sizes this threshold simplifies to $\mathcal{O}(|\beta_b|/k)$, while for small buffer sizes, \hbox{it becomes $\mathcal{O}(k)$}.

Following the coarsening phase, an initial partitioning assigns all non-artificial vertices to blocks using the generalized Fennel algorithm~\cite{HeiStream} with a strict balancing constraint $L_{\max}$.
After initial partitioning, the current solution is transferred to the next finer level by mapping the block assignment of each coarse vertex to its constituent vertices at the finer level.
Subsequently, a local search algorithm is applied at each level, which utilizes the size-constrained label propagation approach employed in the contraction phase but with a modified objective function.
Specifically, when visiting a non-artificial vertex, we reassign it to a neighboring block to maximize the generalized Fennel gain function while strictly adhering to the balancing constraint $L_{\max}$, considering only adjacent blocks in contrast to the initial partitioning which considers all blocks.
This ensures that a single algorithm round remains linear in the current level's size.
Artificial vertices remain stationary during the process, but unlike during coarsening, they are not excluded from label propagation, as they contribute to the generalized Fennel gain function.

\subparagraph*{Faster Initial Partitioning.}
\label{par:Faster Initial Partitioning}
We adopt a modified Fennel function to enhance the initial partitioning step in \texttt{HeiStreamE} and demonstrate that this approach produces a better runtime without \hbox{a decrease in solution quality}.

For small buffer sizes, \texttt{HeiStream} has a linear dependency on $k$ for overall partitioning time, as during initial partitioning each vertex $u \in \beta_b$ is assigned to the block with the highest score among all $k$ blocks. To address this dependency, we adopt Eyubov et al.'s strategy~\cite{freight_paper} from streaming hypergraph partitioning to evaluate scores more efficiently, removing the runtime dependency on $k$.
For the current vertex $u \in \beta_b$, we categorize $k$ blocks $V_i$ into two sets, $S_1$ and $S_2$, where $V_i \in S_1$ if a neighbor of $u$ was assigned to it, and $V_i \in S_2$ otherwise.
This allows us to determine the blocks $V_{max}$ and $V_{max}'$ that respectively maximize Equation (\ref{eq:fenk1}) and Equation (\ref{eq:fenk2}).
Then, the block that maximizes the generalized Fennel function is $\max(V_{max}, V_{max}')$ where 
\begin{equation}
	V_{max} = \argmax_{i \in S_1} \left\{\sum_{v \in V_i \cap N(u)}\omega(u,v) - c(u)f(c(V_i))\right\}
	\label{eq:fenk1}
\end{equation}

\begin{equation}
	V_{max}' = \argmax_{i \in S_2} \left\{- c(u)f(c(V_i))\right\}.
	\label{eq:fenk2}
\end{equation}

As $c(u)$ remains constant, determining the block $V_i \in S_2$ that maximizes Equation (\ref{eq:fenk2}) is equivalent to finding the block $V_i$ minimizing $f(c(V_i)) = \alpha * \gamma * c(V_i)^{\gamma - 1}$, specifically, the $V_i \in S_2$ with the lowest block weight $c(V_i)$. In our scheme, $c(V_i) = c(E_i)$, that is, the number of edges assigned to block $E_i$ of the overall edge partition.
This optimized process, facilitated by a priority queue, reduces the evaluation of blocks for each vertex to those assigned to its neighbors and the minimum weight block overall.
It results in an optimal block $V_i$ for maximizing the generalized Fennel gain function in $\mathcal{O}(d(u)+\log(k))$ time using a binary heap priority queue or $\mathcal{O}(d(u))$ with a bucket priority queue, as suggested in~\cite{freight_paper}, ultimately yielding an overall linear time complexity.
With this enhanced approach, the runtime becomes independent of the parameter $k$.

\subparagraph*{The Parameter $\boldsymbol{\alpha}$.}
\label{par:The Parameter Alpha}
The authors of \texttt{Fennel}~\cite{tsourakakis2014fennel}, define the parameter $\alpha = \sqrt{k}\frac{m}{n^{3/2}}$ for vertex partitioning of an input graph $G = (V, E)$ with $n$ vertices and $m$ edges.
However this choice of $\alpha$ is not directly applicable for edge partitioning in our CSPAC-based model.
If we built a single model $\beta_b$ for the whole graph $G$ at once, it would have $n^* = m$ vertices and a number $m^*$ of edges equal to the number of auxiliary edges in $\beta_b$.
While we know $n^*$ immediately, we cannot directly obtain $m^*$ without visiting all vertices of the whole graph.
Thus, we \hbox{need to estimate $m^*$}.

We thus have three distinct approaches for determining the Fennel parameter $\alpha$.
In the \emph{static} $\alpha$ method, we keep $\alpha$ constant across all batches, setting it to $\alpha = \sqrt{k}\frac{m_{approx}}{n^{3/2}}$, where $n^* = m/2$, $m_{approx} = y * n^*$, and $y$ is a tuning parameter.
For the \emph{batch} $\alpha$ approach, we update $\alpha$ for each batch, calculating it as $\alpha = \sqrt{k}\frac{m_{s}}{n_{s}^{3/2}}$, where $n_{s}$ and $m_s$ are the number of vertices and edges of each CSPAC graph $S^*_b$ respectively.
The \emph{dynamic} $\alpha$ method also updates $\alpha$ per batch. It begins with the \emph{static} $\alpha$ setting and refines $m_{approx}$ in each batch by computing the number of auxiliary edges, determined through counting \hbox{vertices of degree less than or equal to 2}.

\subsection{FreightE}
In addition to \texttt{HeiStreamE}, we present \texttt{FreightE}, a fast streaming edge partitioner that uses streaming hypergraph partitioning to assign blocks to edges on the fly. In general, a hypergraph vertex partitioner can partition edges of an input graph $G$ by first transforming it into its dual hypergraph representation $H$, where each edge of $G$ is a hypervertex, and each vertex of $G$ induces a hyperedge spanning its incident edges. Then, a hypergraph vertex partitioner that assigns the hypervertices of $H$ into $k$ blocks, while optimizing for the connectivity metric, directly provides an edge partition of $G$. The intuition behind this approach is that a hypergraph vertex partitioner that optimizes for the connectivity metric directly optimizes the replication factor of the underlying edge partition~\cite{more_recent_advances_hgp}. In \texttt{FreightE}, we perform the transformation of $G$ into $H$ on the fly as follows. At a given step, we read a vertex of the input graph (METIS format) along with its neighborhood. Each undirected edge in the neighborhood is treated as a unique hypervertex of the hypergraph, and permanently assigned to a block using the \texttt{FREIGHT} streaming hypergraph partitioner~\cite{freight_paper}. This process is repeated until all vertices along with their neighborhoods are visited, at which point each edge of the input graph is assigned to a block.  Following the \texttt{FREIGHT} partitioner, the overall runtime of \texttt{FreightE} is $\mathcal{O}(n+m)$, and memory complexity is $\mathcal{O}(m+k)$. In comparison to \texttt{HeiStreamE}, \texttt{FreightE} is faster, as it does not require the construction of an equivalent CSPAC graph, and requires less memory unless $m \gg n$.

\section{Experimental Evaluation}
\label{sec:Experimental Evaluation}

\subparagraph*{Setup.} 

We implemented \texttt{HeiStreamE} and \texttt{FreightE} inside the \texttt{KaHIP} framework (using C++) and compiled it using gcc 9.3 with full optimization enabled (-O3 flag).
Except for the largest graph instance \texttt{gsh-2015}, all experiments were performed on a single core of a machine consisting of a sixteen-core Intel Xeon Silver 4216 processor running at $2.1$ GHz, $100$ GB of main memory, $16$ MB of L2-Cache, and $22$ MB of L3-Cache running Ubuntu 20.04.1. %
To facilitate algorithms that required greater memory, \texttt{gsh-2015} was run on a single core of an alternate machine, consisting of a 64-core AMD EPYC 7702P Processor containing 1 TB of main memory. 

\subparagraph*{Baselines.}

We compare \texttt{HeiStreamE} and \texttt{FreightE} against the state-of-the-art streaming algorithm \texttt{HDRF}, as well as the re-streaming algorithms \texttt{2PS-HDRF}~\cite{2PS2020} and \texttt{2PS-L}~\cite{psl}, which require three passes over the input graph.
We exclude the following algorithms: \texttt{SNE}~\cite{zhang2017graph}, as it fails to execute for $k > 127$ and is outperformed by \texttt{2PS-HDRF}~\cite{2PS2020}; \texttt{DBH}~\cite{xie2014distributed}, as it ignores past partition assignments and thus has poor solution quality; \texttt{ADWISE} and \texttt{RBSEP} as they have limited global information during streaming, and \texttt{2PS-L} outperforms them~\cite{psl}.

For comparison with competitors, we obtained implementations of \texttt{2PS-HDRF} and \texttt{2PS-L} from their official repository, which also provides an implementation of \texttt{HDRF}.
We configure all competitor algorithms with the optimal settings provided by the authors.
\texttt{2PS-HDRF}, \texttt{2PS-L} and \texttt{HDRF} require a vertex-to-partition table of size $\mathcal{O}(n*k)$, which stores the blocks that each vertex was replicated on. To optimize memory usage, these partitioners can be built with the number of partitions at compile time to only allocate the required amount of memory. 
Thus, we re-compile them with the CMake flag for the number of blocks for each $k$ value in our experiments. We set the \texttt{HDRF} scoring function parameter $\lambda = 1.1$, similar to the authors of \texttt{2PS}~\cite{2PS2020}.
The provided codes for \texttt{2PS-HDRF}, \texttt{2PS-L}, and \texttt{HDRF} in their official repositories set a hard-coded soft limit on the number of partitions to $256$, which we override to test the algorithms with larger $k$ values.

All competitors 
read a binary edge list format with 32-bit vertex IDs.
This allows for faster IO during program execution.
Additionally, all competitor programs offer a converter which can load a graph in the standard edge list format, convert it into the binary edge list format, and write it to memory before proceeding.
For a fair comparison with  \texttt{HeiStreamE} and \texttt{FreightE}, which are capable of reading both text-based (METIS) and binary (adjacency) graphs, we perform experiments with binary graph formats only. Further, we exclude the IO time since the objective of the experiments is to measure the performance of the partitioners and not IO efficiency. Similarly, we exclude the time it takes to convert the graphs into the binary format for all~partitioners.  

\subparagraph*{Instances.}

Our graph instances for experiments are shown in Appendix Table \ref{table:graph} and are sourced from Ref.~\cite{benchmarksfornetworksanalysis,BMSB,BRSLLP,BoVWFI,12m,snap,nr-aaai15}.
All instances evaluated have been used for benchmarking in previous works on graph partitioning.
From these graph instances, we construct three disjoint sets: a tuning set for parameter study experiments, a test set for comparison against state-of-the-art and a set of huge graphs, for which in-memory partitioners ran out of memory on our machine.
We set the number of blocks to $k = \{2^1, 2^2, \dots, 2^{14}\}$ for all experiments except those on huge graphs, for which $k$ values are shown in Table~\ref{table:huge}.
We allow an imbalance of $\varepsilon = 3\%$ for all partitioners.
While streaming, we use the natural order of the vertices in these graphs.

\subparagraph*{Methodology.}

We measure running time, replication factor and memory consumption, i.e., the maximum resident set size for the executed process.
When averaging over all instances, we use the geometric mean to give every instance the same influence on the final score.
Further, we average all results of each algorithm grouped by $k$, to explore performance with increasing $k$ values.
Let the runtime, replication factor or memory consumption be denoted by the score $\sigma_{A}$ for some $k$ partition generated by an algorithm $A$.
We express this score relative to others using the following tools:
\emph{improvement} over an algorithm $B$, computed as a percentage $(\frac{\sigma_A}{\sigma_B} - 1) * 100 \%$ and
\emph{relative} value over an algorithm $B$, computed as $\frac{\sigma_A}{\sigma_B}$.
Additionally, we present performance profiles by Dolan and Mor{\'e}~\cite{pp} to benchmark our algorithms.
These profiles relate the running time (resp. solution quality, memory) of the slower (resp. worse) algorithms to the fastest (resp. best) one on a per-instance basis, rather than grouped by $k$.
Their $x$-axis shows a factor $\tau$ while their $y$-axis shows the percentage of instances for which an algorithm has up to $\tau$ times the running time (resp. solution quality, memory) of the fastest (resp. best) algorithm.

\subsection{Parameter Tuning}
\label{subsec:Tuning}
We tuned the parameters used by \texttt{HeiStreamE}, namely the initial partitioning approach, the selection of the batch graph model $\beta_b$ mode and the choice of Fennel $\alpha$, through experiments run on the Tuning Set (see Appendix Table \ref{table:graph}).
In each experiment we tuned a single parameter with all others constant. 
We ran all tuning experiments on \texttt{HeiStreamE}, with a buffer size of $\delta = \{32\,768, 131\,072, 262\,144\}$. In this section, we describe results with $\delta = 32\,768$ only, as we found that the choice of the best tuning parameters was independent of buffer size. Our results support the use of $k$-independent initial partitioning, the \emph{minimal} mode for $\beta_b$ and batch $\alpha$ for Fennel’s $\alpha$ parameter for subsequent experiments. Other parameters for the multilevel partitioning scheme, specifically the number of rounds of label propagation during coarsening and uncoarsening, and the size of the coarsest graph, align with optimal values from \texttt{HeiStream}~\cite{HeiStream}.

\subparagraph{k-Independent Initial Partitioning.} We evaluated the effect of using our enhanced initial partitioning approach described in Section \ref{par:Faster Initial Partitioning}.
In our implementation, we use a binary heap priority queue to perform increase-key operations in $\mathcal{O}(\log{k})$ time and obtain the block with the smallest weight in $\mathcal{O}(1)$ time.
As a baseline, we use the initial partitioning method in \texttt{HeiStream}~\cite{HeiStream}.
Our enhanced initial partitioning approach is, on average, $1.85\times$ faster than the baseline across all $k$ values, and $4.17\times$ faster for $k \geq 256$. 
The solution quality remains unchanged. 

\subparagraph{Graph Model Mode.} To choose a suitable per-batch graph model mode among the \emph{maximal}, \emph{r-Subset}, and \emph{minimal} modes, we ran comparisons against a baseline configuration using no mode, i.e., one without any artificial vertices or edges representing block assignment decisions.
The \emph{minimal} model mode produces a solution quality improvement of 11.73\% over the baseline of using no mode, which is comparable to the 12.8\% increase achieved by the \emph{maximal} mode.  
All modes have an increased runtime over the baseline. The \emph{minimal} mode is $1.11\times$ slower on average than no mode, but is $1.41\times$ and $1.29\times$ faster than the \emph{maximal} and \emph{r-Subset} mode respectively.
Besides offering a substantial increase in solution quality while being faster than other modes, the \emph{minimal} mode also has a much lower memory overhead compared to the \emph{r-Subset} and \emph{maximal} mode, as we store only one block per \hbox{vertex $v$} (instead of up to $\min\{k, d(v)\}$ blocks).

\subparagraph{Fennel Alpha.} We performed comparisons of the different choices of $\alpha$ among \emph{static}, \emph{batch}, and \emph{dynamic} $\alpha$.
We do not observe a significant difference in runtime between the choices for $\alpha$.
In terms of replication factor, \emph{batch} $\alpha$ provides the best solution quality for a majority of $k$ values, particularly for $k > 32$.
On average, across all instances and all $k$ values, \emph{batch} $\alpha$ produces $0.86\%$ and $3.27\%$ better solution quality than \emph{static} and \emph{dynamic} $\alpha$ respectively.
For $k > 32$, these averages increase to $2.1\%$ and $5.2\%$ respectively.

\subsection{Comparison with State-of-the-Art}
\label{subsec:State-of-the-Art}

We now provide experiments in which we compare \texttt{HeiStreamE} and \texttt{FreightE} against the current state-of-the-art algorithms for (re)streaming edge partitioning, namely, \texttt{HDRF}, \texttt{2PS-HDRF} and \texttt{2PS-L}.
These experiments were performed on the Test Set and the Huge Set of graphs in Appendix Table \ref{table:graph}.
Figure~\ref{fig:overall} gives performance profiles for the Test Set and Table~\ref{table:huge} gives detailed per instance results for instances of the Huge Set.
We distinguish between buffer sizes by defining, for example, \texttt{HeiStreamE(32x)} as \texttt{HeiStreamE} with a buffer size of $32\,768 = 32x$ vertices, where $x = 1\,024$. 

\begin{figure}[t]
	\centering
	
	\includegraphics{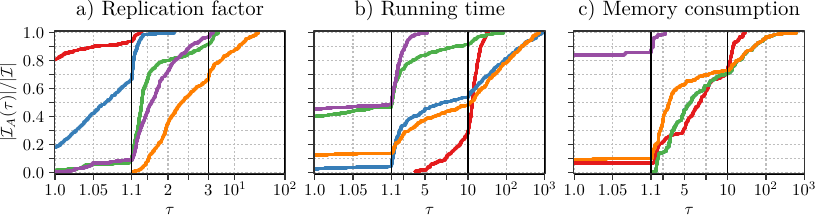}
	\includegraphics{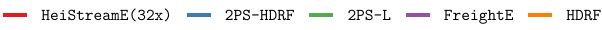}
	
	\caption{Comparison of \texttt{HeiStreamE} and \texttt{FreightE} with \texttt{2PS-HDRF}, \texttt{2PS-L} and \texttt{HDRF} on the Test Set in Appendix Table \ref{table:graph} using performance profiles.
		Let $\mathcal{A}$ be the set of all algorithms, $\mathcal{I}$ the set of instances, and $q_A(I)$ the quality of algorithm $A \in \mathcal{A}$ on instance $I \in \mathcal{I}$. For each algorithm $A$, we plot the fraction of instances $\frac{|\mathcal{I}_A(\tau)|}{|\mathcal{I}|}$ (y-axis) where $\mathcal{I}_A(\tau) := \set{I \in \mathcal{I} | q_A \leq \tau \cdot min_{A' \in \mathcal{A}}q_{A'}(I)}$ and $\tau$ is on the x-axis.
		Includes all $k$ values.
		Note the logarithmic scale in the final third of the plots.
		Memory consumption is measured as the maximum resident set size of the program execution.}
	\label{fig:overall}
\end{figure}

\begin{figure}[t]
	\centering
	
	\includegraphics{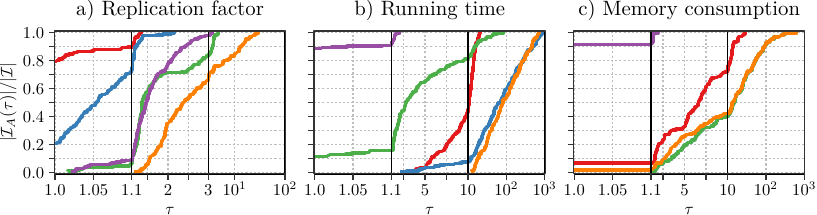}
	\includegraphics{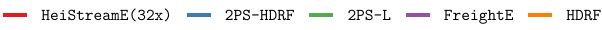}
	
	\caption{Comparison of \texttt{HeiStreamE} and \texttt{FreightE} with \texttt{2PS-HDRF}, \texttt{2PS-L} and \texttt{HDRF} on the Test Set in Appendix Table \ref{table:graph} using performance profiles (see Figure \ref{fig:overall} caption for detailed description). 
		Only includes runs with $k \ge 512$.
		Note the logarithmic scale in the final third of the plots.
		Memory consumption is measured as the maximum resident set size of the program execution.}
	\label{fig:overall_largek}
\end{figure}

\subparagraph{Replication Factor.} \texttt{HeiStreamE(32x)} produces better solution quality than \hbox{state-of-the-art} (re)streaming edge partitioners for all $k$ values. The results from our Test Set demonstrate that \texttt{HeiStreamE(32x)} achieves an average improvement in solution quality of 7.56\% (or 13.69\% when using \texttt{HeiStreamE(256x)}) compared to \texttt{2PS-HDRF}, which produces the next best solution quality. 
As displayed in Figure~\ref{fig:overall}a, \texttt{HeiStreamE(32x)} produces better solution quality than \texttt{2PS-HDRF} in approximately 80\% of all Test Set instances. 
The largest improvement in replication factor that we observed for \texttt{HeiStreamE(32x)} over \texttt{2PS-HDRF} is of $\approx 53\%$ on graph \texttt{circuit5M} for $k \geq 8192$.
Further, \texttt{HeiStreamE(32x)} achieves an average improvement in solution quality of 51.84\% percent over \texttt{2PS-L}, and an average improvement of 202.86\% over \texttt{HDRF}, the only other on the fly streaming algorithm presented here. These results from the Test Set are reflected in our experiments on huge graphs, shown in Table \ref{table:huge}: \texttt{HeiStreamE(32x)} and \texttt{HeiStreamE(256x)} produce the best solution quality in most instances, providing the best solution quality for all $k$ values for five of the six huge graphs. While \texttt{FreightE} produces lower solution quality than \texttt{2PS-L} and \texttt{2PS-HDRF} on most Test Set instances, for $k \geq 512$, \texttt{FreightE}, produces 9.88\% higher solution quality than \texttt{2PS-L} on average. \texttt{FreightE} produces 87.28\% better solution quality than \texttt{HDRF} on average across all Test Set instances and $k$ values. The same trends for \texttt{FreightE} are reflected in our results on huge graphs.

\subparagraph{Runtime.} \texttt{HeiStreamE(32x)} is on average slower than \texttt{2PS-HDRF} and \texttt{HDRF} for $k \leq 256$; however, since its runtime is not linearly dependent on $k$, \texttt{HeiStreamE} is substantially faster than \texttt{2PS-HDRF} and \texttt{HDRF} for higher $k$ values (Figure~\ref{fig:overall_largek}b). On the Test Set, \texttt{HeiStreamE(32x)} is on average $6.7\times$ faster than \texttt{2PS-HDRF} and $8.7\times$ faster than \texttt{HDRF} for $k \geq 512$. Compared to \texttt{2PS-L}, whose time complexity is also independent of $k$, \texttt{HeiStreamE(32x)} is on average slower across all instances. Similarly, in our experiments on huge graphs, \texttt{HeiStreamE(32x)} and \texttt{HeiStreamE(256x)} are faster than \texttt{2PS-HDRF} and \texttt{HDRF} for high $k$ values, but slower than \texttt{2PS-L}. On the Test Set, \texttt{FreightE} is the fastest algorithm among competitors;  it is $1.3\times$ faster than \texttt{2PS-L} which is the next fastest algorithm. Additionally, it is $9.9\times$ faster than \texttt{HDRF} and $8\times$ faster than \texttt{2PS-HDRF} on average across all Test Set instances. In our experiments on huge graphs, \texttt{FreightE} is faster than \texttt{HDRF} and \texttt{2PS-HDRF} on average across all instances, and the fastest algorithm for $k \geq 2\,048$. 

\subparagraph{Memory Consumption.} Since \texttt{HeiStreamE(32x)} uses a buffered streaming approach, it consumes on average more memory than \texttt{2PS-HDRF}, \texttt{2PS-L} and \texttt{HDRF} for $k \leq 256$ on the Test Set; however, since its memory consumption is not asymptotically dependent on \hbox{$m$ or $k$}, \texttt{HeiStreamE(32x)} consumes significantly less memory than \texttt{2PS-HDRF}, \texttt{2PS-L}, and \texttt{HDRF} for higher $k$ values. On average, for \hbox{$k \geq 512$}, \texttt{2PS-HDRF} and \texttt{2PS-L} use $3.0\times$ more memory than \texttt{HeiStreamE(32x)}, and \texttt{HDRF} uses $2.7\times$ more memory than \texttt{HeiStreamE(32x)}. Notably, while \texttt{HeiStreamE(32x)} consumes $1.46\times$ more memory on average than \texttt{2PS-HDRF} and \texttt{2PS-L} for $k \leq 256$ on the Test Set, in three out of six of the huge graphs, it is more memory efficient than them across all $k$ values. On the Test Set, \texttt{FreightE} consumes significantly less memory than all competitors as shown in Figure~\ref{fig:overall}c. It consumes $4.9\times$ less memory than \texttt{HDRF} and $7.2\times$ less memory than \texttt{2PS-L} and \texttt{2PS-HDRF} on average across all instances and all $k$ values. Further, on average, \texttt{FreightE} consumes $16\times$ less memory than \texttt{2PS-L} and \texttt{2PS-HDRF} for $k \geq 512$ (Figure~\ref{fig:overall_largek}c). Since \texttt{FreightE}'s memory consumption is linearly dependent on $m$, it is less memory efficient on the huge graphs, which, like other real-world graphs, tend to have many more edges than vertices. \texttt{FreightE} uses more memory than competitors for $k \leq 256$ in five out of six of the huge graphs, but is more memory efficient at high $k$ due to its memory being asymptotically independent of $k$. In our experiments on huge graphs, \texttt{2PS-HDRF}, \texttt{2PS-L} and \texttt{HDRF} exceed the available memory on the machine at $k \geq 4\,096$ for \hbox{uk-2007-05} and \hbox{webbase-2001}, and at $k \geq 8\,192$ for \hbox{com-friendster}. At $k \geq 16\,384$ they exceed the available memory for all graphs, while \texttt{HeiStreamE} and \texttt{FreightE}'s memory consumption is predictable and consistent across $k$.

\setlength\rotheadsize{1.75cm}
\begin{table}[H]
	
	\caption{\label{table:huge} Results of experiments on the Huge Set in Appendix Table \ref{table:graph}.
		Here, we compare \texttt{HeiStreamE(32x)}, i.e., \texttt{HeiStreamE} with a buffer size of $32 \cdot 1\,024$ (resp. \texttt{HeiStreamE(256x)}), and \texttt{FreightE} with state-of-the-art (re)streaming edge partitioners on huge graph instances, displaying the Replication Factor \emph{RF}, Running Time \emph{RT} [s] and Memory Consumption \emph{Mem} [GB].
		The best solution quality for each instance is \textbf{emboldened}. The missing values indicate when an algorithm exceeded the memory of the machine for all graphs except \texttt{gsh-2015}, where the missing value indicates a timeout, i.e., the algorithm exceeded a runtime limit of 24 hours.
	}
	\resizebox{\linewidth}{!}{%
		\begin{tabular}{lrrrrrrrrrrrrrrrrrrr}
			\multicolumn{1}{c}{$G$} & \multicolumn{1}{c}{$k$} & \multicolumn{3}{c}{\texttt{HeiStreamE(32x)}} & \multicolumn{3}{c}{\texttt{HeiStreamE(256x)}} & \multicolumn{3}{c}{\texttt{FreightE}} & \multicolumn{3}{c}{\texttt{2PS-HDRF}} & \multicolumn{3}{c}{\texttt{2PS-L}} & \multicolumn{3}{c}{\texttt{HDRF}} \\
			& \multicolumn{1}{l}{} & \multicolumn{1}{c}{RF} & \multicolumn{1}{c}{{\cellcolor[rgb]{0.937,0.937,0.937}}RT} & \multicolumn{1}{c}{{\cellcolor[rgb]{0.753,0.753,0.753}}Mem} & \multicolumn{1}{c}{RF} & \multicolumn{1}{c}{{\cellcolor[rgb]{0.937,0.937,0.937}}RT} & \multicolumn{1}{c}{{\cellcolor[rgb]{0.753,0.753,0.753}}Mem} & \multicolumn{1}{c}{RF} & \multicolumn{1}{c}{{\cellcolor[rgb]{0.937,0.937,0.937}}RT} & \multicolumn{1}{c}{{\cellcolor[rgb]{0.753,0.753,0.753}}Mem} & \multicolumn{1}{c}{RF} & \multicolumn{1}{c}{{\cellcolor[rgb]{0.937,0.937,0.937}}RT} & \multicolumn{1}{c}{{\cellcolor[rgb]{0.753,0.753,0.753}}Mem} & \multicolumn{1}{c}{RF} & \multicolumn{1}{c}{{\cellcolor[rgb]{0.937,0.937,0.937}}RT} & \multicolumn{1}{c}{{\cellcolor[rgb]{0.753,0.753,0.753}}Mem} & \multicolumn{1}{c}{RF} & \multicolumn{1}{c}{{\cellcolor[rgb]{0.937,0.937,0.937}}RT} & \multicolumn{1}{c}{{\cellcolor[rgb]{0.753,0.753,0.753}}Mem} \\
			\multicolumn{1}{c}{\multirow{10}{*}{\rotatebox{90}{\textbf{uk-2007-05}}}} & 4 & \textbf{1.05} & {\cellcolor[rgb]{0.937,0.937,0.937}}3385 & {\cellcolor[rgb]{0.753,0.753,0.753}}15.35 & \textbf{1.05} & {\cellcolor[rgb]{0.937,0.937,0.937}}3648 & {\cellcolor[rgb]{0.753,0.753,0.753}}18.59 & 1.76 & {\cellcolor[rgb]{0.937,0.937,0.937}}424 & {\cellcolor[rgb]{0.753,0.753,0.753}}12.72 & 1.12 & {\cellcolor[rgb]{0.937,0.937,0.937}}139 & {\cellcolor[rgb]{0.753,0.753,0.753}}5.92 & 1.40 & {\cellcolor[rgb]{0.937,0.937,0.937}}142 & {\cellcolor[rgb]{0.753,0.753,0.753}}5.92 & 2.74 & {\cellcolor[rgb]{0.937,0.937,0.937}}133 & {\cellcolor[rgb]{0.753,0.753,0.753}}2.35 \\
			\multicolumn{1}{c}{} & 32 & 1.05 & {\cellcolor[rgb]{0.937,0.937,0.937}}3707 & {\cellcolor[rgb]{0.753,0.753,0.753}}15.36 & \textbf{1.04} & {\cellcolor[rgb]{0.937,0.937,0.937}}3610 & {\cellcolor[rgb]{0.753,0.753,0.753}}18.54 & 3.03 & {\cellcolor[rgb]{0.937,0.937,0.937}}427 & {\cellcolor[rgb]{0.753,0.753,0.753}}12.72 & 1.16 & {\cellcolor[rgb]{0.937,0.937,0.937}}283 & {\cellcolor[rgb]{0.753,0.753,0.753}}5.92 & 2.05 & {\cellcolor[rgb]{0.937,0.937,0.937}}149 & {\cellcolor[rgb]{0.753,0.753,0.753}}5.92 & 8.30 & {\cellcolor[rgb]{0.937,0.937,0.937}}718 & {\cellcolor[rgb]{0.753,0.753,0.753}}2.35 \\
			\multicolumn{1}{c}{} & 128 & 1.07 & {\cellcolor[rgb]{0.937,0.937,0.937}}3411 & {\cellcolor[rgb]{0.753,0.753,0.753}}15.37 & \textbf{1.05} & {\cellcolor[rgb]{0.937,0.937,0.937}}3543 & {\cellcolor[rgb]{0.753,0.753,0.753}}18.54 & 3.59 & {\cellcolor[rgb]{0.937,0.937,0.937}}419 & {\cellcolor[rgb]{0.753,0.753,0.753}}12.72 & 1.22 & {\cellcolor[rgb]{0.937,0.937,0.937}}918 & {\cellcolor[rgb]{0.753,0.753,0.753}}6.70 & 2.64 & {\cellcolor[rgb]{0.937,0.937,0.937}}158 & {\cellcolor[rgb]{0.753,0.753,0.753}}6.70 & 12.45 & {\cellcolor[rgb]{0.937,0.937,0.937}}2589 & {\cellcolor[rgb]{0.753,0.753,0.753}}3.92 \\
			\multicolumn{1}{c}{} & 256 & 1.09 & {\cellcolor[rgb]{0.937,0.937,0.937}}3642 & {\cellcolor[rgb]{0.753,0.753,0.753}}15.37 & \textbf{1.06} & {\cellcolor[rgb]{0.937,0.937,0.937}}3684 & {\cellcolor[rgb]{0.753,0.753,0.753}}18.54 & 3.81 & {\cellcolor[rgb]{0.937,0.937,0.937}}398 & {\cellcolor[rgb]{0.753,0.753,0.753}}12.71 & 1.26 & {\cellcolor[rgb]{0.937,0.937,0.937}}1725 & {\cellcolor[rgb]{0.753,0.753,0.753}}8.27 & 3.06 & {\cellcolor[rgb]{0.937,0.937,0.937}}170 & {\cellcolor[rgb]{0.753,0.753,0.753}}8.27 & 14.76 & {\cellcolor[rgb]{0.937,0.937,0.937}}4682 & {\cellcolor[rgb]{0.753,0.753,0.753}}7.06 \\
			\multicolumn{1}{c}{} & 512 & 1.10 & {\cellcolor[rgb]{0.937,0.937,0.937}}3407 & {\cellcolor[rgb]{0.753,0.753,0.753}}16.56 & \textbf{1.06} & {\cellcolor[rgb]{0.937,0.937,0.937}}3718 & {\cellcolor[rgb]{0.753,0.753,0.753}}18.35 & 4.01 & {\cellcolor[rgb]{0.937,0.937,0.937}}373 & {\cellcolor[rgb]{0.753,0.753,0.753}}12.72 & 1.30 & {\cellcolor[rgb]{0.937,0.937,0.937}}3525 & {\cellcolor[rgb]{0.753,0.753,0.753}}14.50 & 3.60 & {\cellcolor[rgb]{0.937,0.937,0.937}}201 & {\cellcolor[rgb]{0.753,0.753,0.753}}14.50 & 17.08 & {\cellcolor[rgb]{0.937,0.937,0.937}}8714 & {\cellcolor[rgb]{0.753,0.753,0.753}}13.32 \\
			\multicolumn{1}{c}{} & 1\,024 & 1.13 & {\cellcolor[rgb]{0.937,0.937,0.937}}3421 & {\cellcolor[rgb]{0.753,0.753,0.753}}16.54 & \textbf{1.08} & {\cellcolor[rgb]{0.937,0.937,0.937}}3685 & {\cellcolor[rgb]{0.753,0.753,0.753}}18.54 & 4.18 & {\cellcolor[rgb]{0.937,0.937,0.937}}371 & {\cellcolor[rgb]{0.753,0.753,0.753}}12.72 & 1.38 & {\cellcolor[rgb]{0.937,0.937,0.937}}7505 & {\cellcolor[rgb]{0.753,0.753,0.753}}27.03 & 4.39 & {\cellcolor[rgb]{0.937,0.937,0.937}}279 & {\cellcolor[rgb]{0.753,0.753,0.753}}27.03 & 19.55 & {\cellcolor[rgb]{0.937,0.937,0.937}}16753 & {\cellcolor[rgb]{0.753,0.753,0.753}}25.86 \\
			\multicolumn{1}{c}{} & 2\,048 & 1.18 & {\cellcolor[rgb]{0.937,0.937,0.937}}3406 & {\cellcolor[rgb]{0.753,0.753,0.753}}15.39 & \textbf{1.11} & {\cellcolor[rgb]{0.937,0.937,0.937}}3551 & {\cellcolor[rgb]{0.753,0.753,0.753}}18.54 & 4.34 & {\cellcolor[rgb]{0.937,0.937,0.937}}371 & {\cellcolor[rgb]{0.753,0.753,0.753}}12.72 & 1.51 & {\cellcolor[rgb]{0.937,0.937,0.937}}17720 & {\cellcolor[rgb]{0.753,0.753,0.753}}52.10 & 5.58 & {\cellcolor[rgb]{0.937,0.937,0.937}}527 & {\cellcolor[rgb]{0.753,0.753,0.753}}52.10 & 22.10 & {\cellcolor[rgb]{0.937,0.937,0.937}}33302 & {\cellcolor[rgb]{0.753,0.753,0.753}}50.93 \\
			\multicolumn{1}{c}{} & 4\,096 & 1.24 & {\cellcolor[rgb]{0.937,0.937,0.937}}3419 & {\cellcolor[rgb]{0.753,0.753,0.753}}16.78 & \textbf{1.16} & {\cellcolor[rgb]{0.937,0.937,0.937}}3538 & {\cellcolor[rgb]{0.753,0.753,0.753}}18.54 & 4.53 & {\cellcolor[rgb]{0.937,0.937,0.937}}370 & {\cellcolor[rgb]{0.753,0.753,0.753}}12.73 & \multicolumn{1}{c}{-} & \multicolumn{1}{c}{{\cellcolor[rgb]{0.937,0.937,0.937}}-} & \multicolumn{1}{c}{{\cellcolor[rgb]{0.753,0.753,0.753}}-} & \multicolumn{1}{c}{-} & \multicolumn{1}{c}{{\cellcolor[rgb]{0.937,0.937,0.937}}-} & \multicolumn{1}{c}{{\cellcolor[rgb]{0.753,0.753,0.753}}-} & \multicolumn{1}{c}{-} & \multicolumn{1}{c}{{\cellcolor[rgb]{0.937,0.937,0.937}}-} & \multicolumn{1}{c}{{\cellcolor[rgb]{0.753,0.753,0.753}}-} \\
			\multicolumn{1}{c}{} & 8\,192 & 1.33 & {\cellcolor[rgb]{0.937,0.937,0.937}}3413 & {\cellcolor[rgb]{0.753,0.753,0.753}}16.94 & \textbf{1.23} & {\cellcolor[rgb]{0.937,0.937,0.937}}3580 & {\cellcolor[rgb]{0.753,0.753,0.753}}18.54 & 4.72 & {\cellcolor[rgb]{0.937,0.937,0.937}}384 & {\cellcolor[rgb]{0.753,0.753,0.753}}12.73 & \multicolumn{1}{c}{-} & \multicolumn{1}{c}{{\cellcolor[rgb]{0.937,0.937,0.937}}-} & \multicolumn{1}{c}{{\cellcolor[rgb]{0.753,0.753,0.753}}-} & \multicolumn{1}{c}{-} & \multicolumn{1}{c}{{\cellcolor[rgb]{0.937,0.937,0.937}}-} & \multicolumn{1}{c}{{\cellcolor[rgb]{0.753,0.753,0.753}}-} & \multicolumn{1}{c}{-} & \multicolumn{1}{c}{{\cellcolor[rgb]{0.937,0.937,0.937}}-} & \multicolumn{1}{c}{{\cellcolor[rgb]{0.753,0.753,0.753}}-} \\
			\multicolumn{1}{c}{} & 16\,384 & 1.50 & {\cellcolor[rgb]{0.937,0.937,0.937}}3464 & {\cellcolor[rgb]{0.753,0.753,0.753}}15.67 & \textbf{1.32} & {\cellcolor[rgb]{0.937,0.937,0.937}}3577 & {\cellcolor[rgb]{0.753,0.753,0.753}}18.54 & 4.97 & {\cellcolor[rgb]{0.937,0.937,0.937}}509 & {\cellcolor[rgb]{0.753,0.753,0.753}}12.73 & \multicolumn{1}{c}{-} & \multicolumn{1}{c}{{\cellcolor[rgb]{0.937,0.937,0.937}}-} & \multicolumn{1}{c}{{\cellcolor[rgb]{0.753,0.753,0.753}}-} & \multicolumn{1}{c}{-} & \multicolumn{1}{c}{{\cellcolor[rgb]{0.937,0.937,0.937}}-} & \multicolumn{1}{c}{{\cellcolor[rgb]{0.753,0.753,0.753}}-} & \multicolumn{1}{c}{-} & \multicolumn{1}{c}{{\cellcolor[rgb]{0.937,0.937,0.937}}-} & \multicolumn{1}{c}{{\cellcolor[rgb]{0.753,0.753,0.753}}-} \\
			\multicolumn{20}{l}{} \\
			\multicolumn{1}{c}{\multirow{10}{*}{\rotatebox{90}{\textbf{com-friendster}}}} & 4 & \textbf{1.62} & {\cellcolor[rgb]{0.937,0.937,0.937}}3356 & {\cellcolor[rgb]{0.753,0.753,0.753}}1.19 & 1.74 & {\cellcolor[rgb]{0.937,0.937,0.937}}7230 & \multicolumn{1}{c}{{\cellcolor[rgb]{0.753,0.753,0.753}}7.08} & 2.83 & {\cellcolor[rgb]{0.937,0.937,0.937}}357 & {\cellcolor[rgb]{0.753,0.753,0.753}}6.98 & 1.88 & {\cellcolor[rgb]{0.937,0.937,0.937}}423 & {\cellcolor[rgb]{0.753,0.753,0.753}}3.72 & 2.08 & {\cellcolor[rgb]{0.937,0.937,0.937}}394 & {\cellcolor[rgb]{0.753,0.753,0.753}}3.72 & 2.35 & {\cellcolor[rgb]{0.937,0.937,0.937}}251 & {\cellcolor[rgb]{0.753,0.753,0.753}}1.47 \\
			& 32 & \textbf{4.90} & {\cellcolor[rgb]{0.937,0.937,0.937}}3410 & {\cellcolor[rgb]{0.753,0.753,0.753}}1.22 & 4.94 & {\cellcolor[rgb]{0.937,0.937,0.937}}7159 & \multicolumn{1}{c}{{\cellcolor[rgb]{0.753,0.753,0.753}}7.33} & 10.71 & {\cellcolor[rgb]{0.937,0.937,0.937}}314 & {\cellcolor[rgb]{0.753,0.753,0.753}}6.98 & 5.21 & {\cellcolor[rgb]{0.937,0.937,0.937}}865 & {\cellcolor[rgb]{0.753,0.753,0.753}}3.72 & 6.99 & {\cellcolor[rgb]{0.937,0.937,0.937}}422 & {\cellcolor[rgb]{0.753,0.753,0.753}}3.72 & 7.97 & {\cellcolor[rgb]{0.937,0.937,0.937}}561 & {\cellcolor[rgb]{0.753,0.753,0.753}}1.47 \\
			& 128 & 10.00 & {\cellcolor[rgb]{0.937,0.937,0.937}}3516 & {\cellcolor[rgb]{0.753,0.753,0.753}}1.23 & 9.53 & {\cellcolor[rgb]{0.937,0.937,0.937}}7370 & \multicolumn{1}{c}{{\cellcolor[rgb]{0.753,0.753,0.753}}7.60} & 18.47 & {\cellcolor[rgb]{0.937,0.937,0.937}}311 & {\cellcolor[rgb]{0.753,0.753,0.753}}6.98 & \textbf{9.00} & {\cellcolor[rgb]{0.937,0.937,0.937}}2159 & {\cellcolor[rgb]{0.753,0.753,0.753}}4.21 & 12.17 & {\cellcolor[rgb]{0.937,0.937,0.937}}526 & {\cellcolor[rgb]{0.753,0.753,0.753}}4.21 & 14.53 & {\cellcolor[rgb]{0.937,0.937,0.937}}1523 & {\cellcolor[rgb]{0.753,0.753,0.753}}2.45 \\
			& 256 & 13.10 & {\cellcolor[rgb]{0.937,0.937,0.937}}3585 & {\cellcolor[rgb]{0.753,0.753,0.753}}1.26 & \textbf{10.87} & {\cellcolor[rgb]{0.937,0.937,0.937}}7530 & \multicolumn{1}{c}{{\cellcolor[rgb]{0.753,0.753,0.753}}7.60} & 19.79 & {\cellcolor[rgb]{0.937,0.937,0.937}}308 & {\cellcolor[rgb]{0.753,0.753,0.753}}6.98 & 11.12 & {\cellcolor[rgb]{0.937,0.937,0.937}}3633 & {\cellcolor[rgb]{0.753,0.753,0.753}}5.18 & 15.05 & {\cellcolor[rgb]{0.937,0.937,0.937}}557 & {\cellcolor[rgb]{0.753,0.753,0.753}}5.18 & 18.00 & {\cellcolor[rgb]{0.937,0.937,0.937}}2611 & {\cellcolor[rgb]{0.753,0.753,0.753}}4.40 \\
			& 512 & 15.92 & {\cellcolor[rgb]{0.937,0.937,0.937}}3681 & {\cellcolor[rgb]{0.753,0.753,0.753}}1.31 & \textbf{11.98} & {\cellcolor[rgb]{0.937,0.937,0.937}}7688 & \multicolumn{1}{c}{{\cellcolor[rgb]{0.753,0.753,0.753}}7.59} & 22.68 & {\cellcolor[rgb]{0.937,0.937,0.937}}306 & {\cellcolor[rgb]{0.753,0.753,0.753}}6.98 & 13.19 & {\cellcolor[rgb]{0.937,0.937,0.937}}6513 & {\cellcolor[rgb]{0.753,0.753,0.753}}9.05 & 17.69 & {\cellcolor[rgb]{0.937,0.937,0.937}}615 & {\cellcolor[rgb]{0.753,0.753,0.753}}9.05 & 21.16 & {\cellcolor[rgb]{0.937,0.937,0.937}}4743 & {\cellcolor[rgb]{0.753,0.753,0.753}}8.31 \\
			& 1\,024 & 18.05 & {\cellcolor[rgb]{0.937,0.937,0.937}}3742 & {\cellcolor[rgb]{0.753,0.753,0.753}}1.31 & \textbf{12.90} & {\cellcolor[rgb]{0.937,0.937,0.937}}7828 & \multicolumn{1}{c}{{\cellcolor[rgb]{0.753,0.753,0.753}}7.76} & 24.89 & {\cellcolor[rgb]{0.937,0.937,0.937}}312 & {\cellcolor[rgb]{0.753,0.753,0.753}}6.98 & 15.16 & {\cellcolor[rgb]{0.937,0.937,0.937}}12226 & {\cellcolor[rgb]{0.753,0.753,0.753}}16.87 & 20.03 & {\cellcolor[rgb]{0.937,0.937,0.937}}702 & {\cellcolor[rgb]{0.753,0.753,0.753}}16.87 & 23.78 & {\cellcolor[rgb]{0.937,0.937,0.937}}9173 & {\cellcolor[rgb]{0.753,0.753,0.753}}16.14 \\
			& 2\,048 & 19.00 & {\cellcolor[rgb]{0.937,0.937,0.937}}3916 & {\cellcolor[rgb]{0.753,0.753,0.753}}1.30 & \textbf{13.91} & {\cellcolor[rgb]{0.937,0.937,0.937}}7929 & \multicolumn{1}{c}{{\cellcolor[rgb]{0.753,0.753,0.753}}7.73} & 25.97 & {\cellcolor[rgb]{0.937,0.937,0.937}}339 & {\cellcolor[rgb]{0.753,0.753,0.753}}6.98 & 16.95 & {\cellcolor[rgb]{0.937,0.937,0.937}}23536 & {\cellcolor[rgb]{0.753,0.753,0.753}}32.51 & 21.87 & {\cellcolor[rgb]{0.937,0.937,0.937}}858 & {\cellcolor[rgb]{0.753,0.753,0.753}}32.51 & 25.76 & {\cellcolor[rgb]{0.937,0.937,0.937}}17678 & {\cellcolor[rgb]{0.753,0.753,0.753}}31.78 \\
			& 4\,096 & 19.77 & {\cellcolor[rgb]{0.937,0.937,0.937}}4024 & {\cellcolor[rgb]{0.753,0.753,0.753}}1.35 & \textbf{15.35} & {\cellcolor[rgb]{0.937,0.937,0.937}}8243 & \multicolumn{1}{c}{{\cellcolor[rgb]{0.753,0.753,0.753}}7.84} & 26.56 & {\cellcolor[rgb]{0.937,0.937,0.937}}358 & {\cellcolor[rgb]{0.753,0.753,0.753}}6.98 & 18.52 & {\cellcolor[rgb]{0.937,0.937,0.937}}45773 & {\cellcolor[rgb]{0.753,0.753,0.753}}63.80 & 23.27 & {\cellcolor[rgb]{0.937,0.937,0.937}}1161 & {\cellcolor[rgb]{0.753,0.753,0.753}}63.80 & 27.15 & {\cellcolor[rgb]{0.937,0.937,0.937}}33753 & {\cellcolor[rgb]{0.753,0.753,0.753}}63.06 \\
			& 8\,192 & 20.26 & {\cellcolor[rgb]{0.937,0.937,0.937}}4154 & {\cellcolor[rgb]{0.753,0.753,0.753}}1.36 & \textbf{16.90} & {\cellcolor[rgb]{0.937,0.937,0.937}}8325 & \multicolumn{1}{c}{{\cellcolor[rgb]{0.753,0.753,0.753}}8.10} & 26.86 & {\cellcolor[rgb]{0.937,0.937,0.937}}404 & {\cellcolor[rgb]{0.753,0.753,0.753}}6.98 & \multicolumn{1}{c}{-} & \multicolumn{1}{c}{{\cellcolor[rgb]{0.937,0.937,0.937}}-} & \multicolumn{1}{c}{{\cellcolor[rgb]{0.753,0.753,0.753}}-} & \multicolumn{1}{c}{-} & \multicolumn{1}{c}{{\cellcolor[rgb]{0.937,0.937,0.937}}-} & \multicolumn{1}{c}{{\cellcolor[rgb]{0.753,0.753,0.753}}-} & \multicolumn{1}{c}{-} & \multicolumn{1}{c}{{\cellcolor[rgb]{0.937,0.937,0.937}}-} & \multicolumn{1}{c}{{\cellcolor[rgb]{0.753,0.753,0.753}}-} \\
			& 16\,384 & 21.10 & {\cellcolor[rgb]{0.937,0.937,0.937}}4339 & {\cellcolor[rgb]{0.753,0.753,0.753}}1.33 & \textbf{19.16} & {\cellcolor[rgb]{0.937,0.937,0.937}}8782 & \multicolumn{1}{c}{{\cellcolor[rgb]{0.753,0.753,0.753}}8.55} & 26.90 & {\cellcolor[rgb]{0.937,0.937,0.937}}479 & {\cellcolor[rgb]{0.753,0.753,0.753}}6.98 & \multicolumn{1}{c}{-} & \multicolumn{1}{c}{{\cellcolor[rgb]{0.937,0.937,0.937}}-} & \multicolumn{1}{c}{{\cellcolor[rgb]{0.753,0.753,0.753}}-} & \multicolumn{1}{c}{-} & \multicolumn{1}{c}{{\cellcolor[rgb]{0.937,0.937,0.937}}-} & \multicolumn{1}{c}{{\cellcolor[rgb]{0.753,0.753,0.753}}-} & \multicolumn{1}{c}{-} & \multicolumn{1}{c}{{\cellcolor[rgb]{0.937,0.937,0.937}}-} & \multicolumn{1}{c}{{\cellcolor[rgb]{0.753,0.753,0.753}}-} \\
			\multicolumn{20}{l}{} \\
			\multicolumn{1}{c}{\multirow{10}{*}{\rotatebox{90}{\textbf{it-2004}}}} & 4 & \textbf{1.10} & {\cellcolor[rgb]{0.937,0.937,0.937}}999 & {\cellcolor[rgb]{0.753,0.753,0.753}}11.65 & \textbf{1.10} & {\cellcolor[rgb]{0.937,0.937,0.937}}1019 & {\cellcolor[rgb]{0.753,0.753,0.753}}17.83 & 1.73 & {\cellcolor[rgb]{0.937,0.937,0.937}}131 & {\cellcolor[rgb]{0.753,0.753,0.753}}4.03 & 1.15 & {\cellcolor[rgb]{0.937,0.937,0.937}}46 & {\cellcolor[rgb]{0.753,0.753,0.753}}2.54 & 1.45 & {\cellcolor[rgb]{0.937,0.937,0.937}}47 & {\cellcolor[rgb]{0.753,0.753,0.753}}2.54 & 2.65 & {\cellcolor[rgb]{0.937,0.937,0.937}}43 & {\cellcolor[rgb]{0.753,0.753,0.753}}0.93 \\
			& 32 & \textbf{1.09} & {\cellcolor[rgb]{0.937,0.937,0.937}}1000 & {\cellcolor[rgb]{0.753,0.753,0.753}}11.23 & \textbf{1.09} & {\cellcolor[rgb]{0.937,0.937,0.937}}1021 & {\cellcolor[rgb]{0.753,0.753,0.753}}17.70 & 2.74 & {\cellcolor[rgb]{0.937,0.937,0.937}}116 & {\cellcolor[rgb]{0.753,0.753,0.753}}4.03 & 1.21 & {\cellcolor[rgb]{0.937,0.937,0.937}}91 & {\cellcolor[rgb]{0.753,0.753,0.753}}2.54 & 2.12 & {\cellcolor[rgb]{0.937,0.937,0.937}}49 & {\cellcolor[rgb]{0.753,0.753,0.753}}2.54 & 7.75 & {\cellcolor[rgb]{0.937,0.937,0.937}}229 & {\cellcolor[rgb]{0.753,0.753,0.753}}0.93 \\
			& 128 & 1.11 & {\cellcolor[rgb]{0.937,0.937,0.937}}1006 & {\cellcolor[rgb]{0.753,0.753,0.753}}11.22 & \textbf{1.09} & {\cellcolor[rgb]{0.937,0.937,0.937}}1021 & {\cellcolor[rgb]{0.753,0.753,0.753}}18.79 & 3.20 & {\cellcolor[rgb]{0.937,0.937,0.937}}116 & {\cellcolor[rgb]{0.753,0.753,0.753}}4.03 & 1.27 & {\cellcolor[rgb]{0.937,0.937,0.937}}271 & {\cellcolor[rgb]{0.753,0.753,0.753}}2.85 & 2.40 & {\cellcolor[rgb]{0.937,0.937,0.937}}52 & {\cellcolor[rgb]{0.753,0.753,0.753}}2.85 & 11.70 & {\cellcolor[rgb]{0.937,0.937,0.937}}787 & {\cellcolor[rgb]{0.753,0.753,0.753}}1.54 \\
			& 256 & 1.12 & {\cellcolor[rgb]{0.937,0.937,0.937}}1001 & {\cellcolor[rgb]{0.753,0.753,0.753}}11.22 & \textbf{1.09} & {\cellcolor[rgb]{0.937,0.937,0.937}}1027 & {\cellcolor[rgb]{0.753,0.753,0.753}}18.01 & 3.38 & {\cellcolor[rgb]{0.937,0.937,0.937}}117 & {\cellcolor[rgb]{0.753,0.753,0.753}}4.03 & 1.31 & {\cellcolor[rgb]{0.937,0.937,0.937}}514 & {\cellcolor[rgb]{0.753,0.753,0.753}}3.47 & 2.52 & {\cellcolor[rgb]{0.937,0.937,0.937}}55 & {\cellcolor[rgb]{0.753,0.753,0.753}}3.47 & 13.85 & {\cellcolor[rgb]{0.937,0.937,0.937}}1425 & {\cellcolor[rgb]{0.753,0.753,0.753}}2.77 \\
			& 512 & 1.14 & {\cellcolor[rgb]{0.937,0.937,0.937}}1008 & {\cellcolor[rgb]{0.753,0.753,0.753}}11.85 & \textbf{1.10} & {\cellcolor[rgb]{0.937,0.937,0.937}}1007 & {\cellcolor[rgb]{0.753,0.753,0.753}}17.90 & 3.53 & {\cellcolor[rgb]{0.937,0.937,0.937}}117 & {\cellcolor[rgb]{0.753,0.753,0.753}}4.03 & 1.35 & {\cellcolor[rgb]{0.937,0.937,0.937}}1033 & {\cellcolor[rgb]{0.753,0.753,0.753}}5.70 & 2.98 & {\cellcolor[rgb]{0.937,0.937,0.937}}66 & {\cellcolor[rgb]{0.753,0.753,0.753}}5.70 & 16.15 & {\cellcolor[rgb]{0.937,0.937,0.937}}2684 & {\cellcolor[rgb]{0.753,0.753,0.753}}5.23 \\
			& 1\,024 & 1.17 & {\cellcolor[rgb]{0.937,0.937,0.937}}1002 & {\cellcolor[rgb]{0.753,0.753,0.753}}11.68 & \textbf{1.12} & {\cellcolor[rgb]{0.937,0.937,0.937}}1013 & {\cellcolor[rgb]{0.753,0.753,0.753}}17.57 & 3.67 & {\cellcolor[rgb]{0.937,0.937,0.937}}117 & {\cellcolor[rgb]{0.753,0.753,0.753}}4.03 & 1.39 & {\cellcolor[rgb]{0.937,0.937,0.937}}2132 & {\cellcolor[rgb]{0.753,0.753,0.753}}10.62 & 3.46 & {\cellcolor[rgb]{0.937,0.937,0.937}}92 & {\cellcolor[rgb]{0.753,0.753,0.753}}10.62 & 18.51 & {\cellcolor[rgb]{0.937,0.937,0.937}}5237 & {\cellcolor[rgb]{0.753,0.753,0.753}}10.16 \\
			& 2\,048 & 1.21 & {\cellcolor[rgb]{0.937,0.937,0.937}}1000 & {\cellcolor[rgb]{0.753,0.753,0.753}}11.23 & \textbf{1.16} & {\cellcolor[rgb]{0.937,0.937,0.937}}999 & {\cellcolor[rgb]{0.753,0.753,0.753}}17.86 & 3.78 & {\cellcolor[rgb]{0.937,0.937,0.937}}116 & {\cellcolor[rgb]{0.753,0.753,0.753}}4.03 & 1.51 & {\cellcolor[rgb]{0.937,0.937,0.937}}5127 & {\cellcolor[rgb]{0.753,0.753,0.753}}20.46 & 4.33 & {\cellcolor[rgb]{0.937,0.937,0.937}}157 & {\cellcolor[rgb]{0.753,0.753,0.753}}20.46 & 20.82 & {\cellcolor[rgb]{0.937,0.937,0.937}}10452 & {\cellcolor[rgb]{0.753,0.753,0.753}}20.00 \\
			& 4\,096 & 1.27 & {\cellcolor[rgb]{0.937,0.937,0.937}}1013 & {\cellcolor[rgb]{0.753,0.753,0.753}}11.87 & \textbf{1.21} & {\cellcolor[rgb]{0.937,0.937,0.937}}1031 & {\cellcolor[rgb]{0.753,0.753,0.753}}19.47 & 3.89 & {\cellcolor[rgb]{0.937,0.937,0.937}}118 & {\cellcolor[rgb]{0.753,0.753,0.753}}4.03 & 1.68 & {\cellcolor[rgb]{0.937,0.937,0.937}}13010 & {\cellcolor[rgb]{0.753,0.753,0.753}}40.15 & 5.89 & {\cellcolor[rgb]{0.937,0.937,0.937}}356 & {\cellcolor[rgb]{0.753,0.753,0.753}}40.15 & 22.94 & {\cellcolor[rgb]{0.937,0.937,0.937}}20717 & {\cellcolor[rgb]{0.753,0.753,0.753}}39.69 \\
			& 8\,192 & 1.38 & {\cellcolor[rgb]{0.937,0.937,0.937}}1020 & {\cellcolor[rgb]{0.753,0.753,0.753}}11.88 & \textbf{1.30} & {\cellcolor[rgb]{0.937,0.937,0.937}}1033 & {\cellcolor[rgb]{0.753,0.753,0.753}}18.08 & 4.02 & {\cellcolor[rgb]{0.937,0.937,0.937}}116 & {\cellcolor[rgb]{0.753,0.753,0.753}}4.03 & 1.90 & {\cellcolor[rgb]{0.937,0.937,0.937}}31776 & {\cellcolor[rgb]{0.753,0.753,0.753}}79.53 & 7.44 & {\cellcolor[rgb]{0.937,0.937,0.937}}817 & {\cellcolor[rgb]{0.753,0.753,0.753}}79.53 & 24.80 & {\cellcolor[rgb]{0.937,0.937,0.937}}40267 & {\cellcolor[rgb]{0.753,0.753,0.753}}79.07 \\
			& 16\,384 & 1.53 & {\cellcolor[rgb]{0.937,0.937,0.937}}1032 & {\cellcolor[rgb]{0.753,0.753,0.753}}11.99 & \textbf{1.40} & {\cellcolor[rgb]{0.937,0.937,0.937}}1050 & {\cellcolor[rgb]{0.753,0.753,0.753}}18.09 & 4.18 & {\cellcolor[rgb]{0.937,0.937,0.937}}151 & {\cellcolor[rgb]{0.753,0.753,0.753}}4.04 & \multicolumn{1}{c}{-} & \multicolumn{1}{c}{{\cellcolor[rgb]{0.937,0.937,0.937}}-} & \multicolumn{1}{c}{{\cellcolor[rgb]{0.753,0.753,0.753}}-} & \multicolumn{1}{c}{-} & \multicolumn{1}{c}{{\cellcolor[rgb]{0.937,0.937,0.937}}-} & \multicolumn{1}{c}{{\cellcolor[rgb]{0.753,0.753,0.753}}-} & \multicolumn{1}{c}{-} & \multicolumn{1}{c}{{\cellcolor[rgb]{0.937,0.937,0.937}}-} & \multicolumn{1}{c}{{\cellcolor[rgb]{0.753,0.753,0.753}}-} \\
			\multicolumn{20}{l}{} \\
			\multicolumn{1}{c}{\multirow{10}{*}{\rotatebox{90}{\textbf{sk-2005}}}} & 4 & \textbf{1.14} & {\cellcolor[rgb]{0.937,0.937,0.937}}2041 & {\cellcolor[rgb]{0.753,0.753,0.753}}12.97 & 1.16 & {\cellcolor[rgb]{0.937,0.937,0.937}}2126 & {\cellcolor[rgb]{0.753,0.753,0.753}}18.31 & 2.18 & {\cellcolor[rgb]{0.937,0.937,0.937}}232 & {\cellcolor[rgb]{0.753,0.753,0.753}}7.17 & 1.17 & {\cellcolor[rgb]{0.937,0.937,0.937}}85 & {\cellcolor[rgb]{0.753,0.753,0.753}}2.94 & 1.30 & {\cellcolor[rgb]{0.937,0.937,0.937}}88 & {\cellcolor[rgb]{0.753,0.753,0.753}}2.94 & 2.93 & {\cellcolor[rgb]{0.937,0.937,0.937}}73 & {\cellcolor[rgb]{0.753,0.753,0.753}}1.14 \\
			& 32 & 1.25 & {\cellcolor[rgb]{0.937,0.937,0.937}}2036 & {\cellcolor[rgb]{0.753,0.753,0.753}}13.25 & \textbf{1.18} & {\cellcolor[rgb]{0.937,0.937,0.937}}2143 & {\cellcolor[rgb]{0.753,0.753,0.753}}18.32 & 4.11 & {\cellcolor[rgb]{0.937,0.937,0.937}}236 & {\cellcolor[rgb]{0.753,0.753,0.753}}7.17 & 1.25 & {\cellcolor[rgb]{0.937,0.937,0.937}}194 & {\cellcolor[rgb]{0.753,0.753,0.753}}2.94 & 3.06 & {\cellcolor[rgb]{0.937,0.937,0.937}}91 & {\cellcolor[rgb]{0.753,0.753,0.753}}2.94 & 9.78 & {\cellcolor[rgb]{0.937,0.937,0.937}}391 & {\cellcolor[rgb]{0.753,0.753,0.753}}1.14 \\
			& 128 & 1.31 & {\cellcolor[rgb]{0.937,0.937,0.937}}2028 & {\cellcolor[rgb]{0.753,0.753,0.753}}13.02 & \textbf{1.19} & {\cellcolor[rgb]{0.937,0.937,0.937}}2147 & {\cellcolor[rgb]{0.753,0.753,0.753}}19.27 & 4.90 & {\cellcolor[rgb]{0.937,0.937,0.937}}202 & {\cellcolor[rgb]{0.753,0.753,0.753}}7.17 & 1.33 & {\cellcolor[rgb]{0.937,0.937,0.937}}630 & {\cellcolor[rgb]{0.753,0.753,0.753}}3.31 & 4.67 & {\cellcolor[rgb]{0.937,0.937,0.937}}102 & {\cellcolor[rgb]{0.753,0.753,0.753}}3.31 & 14.13 & {\cellcolor[rgb]{0.937,0.937,0.937}}1459 & {\cellcolor[rgb]{0.753,0.753,0.753}}1.89 \\
			& 256 & 1.39 & {\cellcolor[rgb]{0.937,0.937,0.937}}2047 & {\cellcolor[rgb]{0.753,0.753,0.753}}13.10 & \textbf{1.22} & {\cellcolor[rgb]{0.937,0.937,0.937}}2158 & {\cellcolor[rgb]{0.753,0.753,0.753}}19.14 & 5.16 & {\cellcolor[rgb]{0.937,0.937,0.937}}203 & {\cellcolor[rgb]{0.753,0.753,0.753}}7.17 & 1.38 & {\cellcolor[rgb]{0.937,0.937,0.937}}1254 & {\cellcolor[rgb]{0.753,0.753,0.753}}4.07 & 5.65 & {\cellcolor[rgb]{0.937,0.937,0.937}}120 & {\cellcolor[rgb]{0.753,0.753,0.753}}4.07 & 16.20 & {\cellcolor[rgb]{0.937,0.937,0.937}}2636 & {\cellcolor[rgb]{0.753,0.753,0.753}}3.40 \\
			& 512 & 1.46 & {\cellcolor[rgb]{0.937,0.937,0.937}}2052 & {\cellcolor[rgb]{0.753,0.753,0.753}}13.09 & \textbf{1.27} & {\cellcolor[rgb]{0.937,0.937,0.937}}2140 & {\cellcolor[rgb]{0.753,0.753,0.753}}18.35 & 5.42 & {\cellcolor[rgb]{0.937,0.937,0.937}}203 & {\cellcolor[rgb]{0.753,0.753,0.753}}7.17 & 1.51 & {\cellcolor[rgb]{0.937,0.937,0.937}}2743 & {\cellcolor[rgb]{0.753,0.753,0.753}}6.98 & 7.04 & {\cellcolor[rgb]{0.937,0.937,0.937}}159 & {\cellcolor[rgb]{0.753,0.753,0.753}}6.98 & 18.40 & {\cellcolor[rgb]{0.937,0.937,0.937}}4836 & {\cellcolor[rgb]{0.753,0.753,0.753}}6.42 \\
			& 1\,024 & 1.55 & {\cellcolor[rgb]{0.937,0.937,0.937}}2035 & {\cellcolor[rgb]{0.753,0.753,0.753}}13.35 & \textbf{1.33} & {\cellcolor[rgb]{0.937,0.937,0.937}}2153 & {\cellcolor[rgb]{0.753,0.753,0.753}}18.35 & 5.64 & {\cellcolor[rgb]{0.937,0.937,0.937}}203 & {\cellcolor[rgb]{0.753,0.753,0.753}}7.17 & 1.72 & {\cellcolor[rgb]{0.937,0.937,0.937}}6358 & {\cellcolor[rgb]{0.753,0.753,0.753}}13.02 & 8.93 & {\cellcolor[rgb]{0.937,0.937,0.937}}261 & {\cellcolor[rgb]{0.753,0.753,0.753}}13.02 & 20.69 & {\cellcolor[rgb]{0.937,0.937,0.937}}9226 & {\cellcolor[rgb]{0.753,0.753,0.753}}12.45 \\
			& 2\,048 & 1.63 & {\cellcolor[rgb]{0.937,0.937,0.937}}2064 & {\cellcolor[rgb]{0.753,0.753,0.753}}13.15 & \textbf{1.42} & {\cellcolor[rgb]{0.937,0.937,0.937}}2145 & {\cellcolor[rgb]{0.753,0.753,0.753}}18.43 & 5.89 & {\cellcolor[rgb]{0.937,0.937,0.937}}205 & {\cellcolor[rgb]{0.753,0.753,0.753}}7.17 & 1.99 & {\cellcolor[rgb]{0.937,0.937,0.937}}14369 & {\cellcolor[rgb]{0.753,0.753,0.753}}25.09 & 9.79 & {\cellcolor[rgb]{0.937,0.937,0.937}}452 & {\cellcolor[rgb]{0.753,0.753,0.753}}25.09 & 23.01 & {\cellcolor[rgb]{0.937,0.937,0.937}}18224 & {\cellcolor[rgb]{0.753,0.753,0.753}}24.53 \\
			& 4\,096 & 1.75 & {\cellcolor[rgb]{0.937,0.937,0.937}}2076 & {\cellcolor[rgb]{0.753,0.753,0.753}}13.27 & \textbf{1.51} & {\cellcolor[rgb]{0.937,0.937,0.937}}2168 & {\cellcolor[rgb]{0.753,0.753,0.753}}18.37 & 6.22 & {\cellcolor[rgb]{0.937,0.937,0.937}}207 & {\cellcolor[rgb]{0.753,0.753,0.753}}7.17 & 2.31 & {\cellcolor[rgb]{0.937,0.937,0.937}}31775 & {\cellcolor[rgb]{0.753,0.753,0.753}}49.24 & 11.27 & {\cellcolor[rgb]{0.937,0.937,0.937}}997 & {\cellcolor[rgb]{0.753,0.753,0.753}}49.24 & 25.39 & {\cellcolor[rgb]{0.937,0.937,0.937}}36426 & {\cellcolor[rgb]{0.753,0.753,0.753}}48.67 \\
			& 8\,192 & 1.86 & {\cellcolor[rgb]{0.937,0.937,0.937}}2095 & {\cellcolor[rgb]{0.753,0.753,0.753}}13.56 & \textbf{1.62} & {\cellcolor[rgb]{0.937,0.937,0.937}}2159 & {\cellcolor[rgb]{0.753,0.753,0.753}}18.38 & 6.64 & {\cellcolor[rgb]{0.937,0.937,0.937}}211 & {\cellcolor[rgb]{0.753,0.753,0.753}}7.17 & 2.67 & {\cellcolor[rgb]{0.937,0.937,0.937}}70438 & {\cellcolor[rgb]{0.753,0.753,0.753}}91.45 & 12.34 & {\cellcolor[rgb]{0.937,0.937,0.937}}2051 & {\cellcolor[rgb]{0.753,0.753,0.753}}91.46 & 27.95 & {\cellcolor[rgb]{0.937,0.937,0.937}}72800 & {\cellcolor[rgb]{0.753,0.753,0.753}}91.44 \\
			& 16\,384 & 2.06 & {\cellcolor[rgb]{0.937,0.937,0.937}}2155 & {\cellcolor[rgb]{0.753,0.753,0.753}}13.09 & \textbf{1.78} & {\cellcolor[rgb]{0.937,0.937,0.937}}2185 & {\cellcolor[rgb]{0.753,0.753,0.753}}18.40 & 7.21 & {\cellcolor[rgb]{0.937,0.937,0.937}}262 & {\cellcolor[rgb]{0.753,0.753,0.753}}7.18 & \multicolumn{1}{c}{-} & \multicolumn{1}{c}{{\cellcolor[rgb]{0.937,0.937,0.937}}-} & \multicolumn{1}{c}{{\cellcolor[rgb]{0.753,0.753,0.753}}-} & \multicolumn{1}{c}{-} & \multicolumn{1}{c}{{\cellcolor[rgb]{0.937,0.937,0.937}}-} & \multicolumn{1}{c}{{\cellcolor[rgb]{0.753,0.753,0.753}}-} & \multicolumn{1}{c}{-} & \multicolumn{1}{c}{{\cellcolor[rgb]{0.937,0.937,0.937}}-} & \multicolumn{1}{c}{{\cellcolor[rgb]{0.753,0.753,0.753}}-} \\
			\multicolumn{20}{l}{} \\
			\multicolumn{1}{c}{\multirow{10}{*}{\rotatebox{90}{\textbf{webbase-2001}}}} & 4 & 1.06 & {\cellcolor[rgb]{0.937,0.937,0.937}}770 & {\cellcolor[rgb]{0.753,0.753,0.753}}1.33 & \textbf{1.05} & {\cellcolor[rgb]{0.937,0.937,0.937}}782 & {\cellcolor[rgb]{0.753,0.753,0.753}}2.20 & 1.44 & {\cellcolor[rgb]{0.937,0.937,0.937}}140 & {\cellcolor[rgb]{0.753,0.753,0.753}}3.65 & 1.10 & {\cellcolor[rgb]{0.937,0.937,0.937}}51 & {\cellcolor[rgb]{0.753,0.753,0.753}}6.32 & 1.34 & {\cellcolor[rgb]{0.937,0.937,0.937}}50 & {\cellcolor[rgb]{0.753,0.753,0.753}}6.32 & 2.17 & {\cellcolor[rgb]{0.937,0.937,0.937}}42 & {\cellcolor[rgb]{0.753,0.753,0.753}}2.59 \\
			& 32 & \textbf{1.08} & {\cellcolor[rgb]{0.937,0.937,0.937}}785 & {\cellcolor[rgb]{0.753,0.753,0.753}}1.31 & 1.09 & {\cellcolor[rgb]{0.937,0.937,0.937}}790 & {\cellcolor[rgb]{0.753,0.753,0.753}}2.20 & 1.96 & {\cellcolor[rgb]{0.937,0.937,0.937}}112 & {\cellcolor[rgb]{0.753,0.753,0.753}}3.65 & 1.16 & {\cellcolor[rgb]{0.937,0.937,0.937}}91 & {\cellcolor[rgb]{0.753,0.753,0.753}}6.32 & 1.61 & {\cellcolor[rgb]{0.937,0.937,0.937}}53 & {\cellcolor[rgb]{0.753,0.753,0.753}}6.32 & 4.49 & {\cellcolor[rgb]{0.937,0.937,0.937}}192 & {\cellcolor[rgb]{0.753,0.753,0.753}}2.59 \\
			& 128 & \textbf{1.09} & {\cellcolor[rgb]{0.937,0.937,0.937}}774 & {\cellcolor[rgb]{0.753,0.753,0.753}}1.31 & \textbf{1.09} & {\cellcolor[rgb]{0.937,0.937,0.937}}789 & {\cellcolor[rgb]{0.753,0.753,0.753}}2.25 & 2.13 & {\cellcolor[rgb]{0.937,0.937,0.937}}113 & {\cellcolor[rgb]{0.753,0.753,0.753}}3.65 & 1.20 & {\cellcolor[rgb]{0.937,0.937,0.937}}245 & {\cellcolor[rgb]{0.753,0.753,0.753}}7.18 & 1.68 & {\cellcolor[rgb]{0.937,0.937,0.937}}57 & {\cellcolor[rgb]{0.753,0.753,0.753}}7.18 & 5.96 & {\cellcolor[rgb]{0.937,0.937,0.937}}610 & {\cellcolor[rgb]{0.753,0.753,0.753}}4.31 \\
			& 256 & \textbf{1.09} & {\cellcolor[rgb]{0.937,0.937,0.937}}782 & {\cellcolor[rgb]{0.753,0.753,0.753}}1.38 & \textbf{1.09} & {\cellcolor[rgb]{0.937,0.937,0.937}}799 & {\cellcolor[rgb]{0.753,0.753,0.753}}2.21 & 2.18 & {\cellcolor[rgb]{0.937,0.937,0.937}}112 & {\cellcolor[rgb]{0.753,0.753,0.753}}3.65 & 1.22 & {\cellcolor[rgb]{0.937,0.937,0.937}}439 & {\cellcolor[rgb]{0.753,0.753,0.753}}9.04 & 1.68 & {\cellcolor[rgb]{0.937,0.937,0.937}}61 & {\cellcolor[rgb]{0.753,0.753,0.753}}9.04 & 6.66 & {\cellcolor[rgb]{0.937,0.937,0.937}}1113 & {\cellcolor[rgb]{0.753,0.753,0.753}}7.75 \\
			& 512 & \textbf{1.09} & {\cellcolor[rgb]{0.937,0.937,0.937}}790 & {\cellcolor[rgb]{0.753,0.753,0.753}}1.34 & \textbf{1.09} & {\cellcolor[rgb]{0.937,0.937,0.937}}797 & {\cellcolor[rgb]{0.753,0.753,0.753}}2.19 & 2.21 & {\cellcolor[rgb]{0.937,0.937,0.937}}111 & {\cellcolor[rgb]{0.753,0.753,0.753}}3.65 & 1.24 & {\cellcolor[rgb]{0.937,0.937,0.937}}817 & {\cellcolor[rgb]{0.753,0.753,0.753}}15.93 & 1.66 & {\cellcolor[rgb]{0.937,0.937,0.937}}67 & {\cellcolor[rgb]{0.753,0.753,0.753}}15.93 & 7.33 & {\cellcolor[rgb]{0.937,0.937,0.937}}2117 & {\cellcolor[rgb]{0.753,0.753,0.753}}14.64 \\
			& 1\,024 & 1.10 & {\cellcolor[rgb]{0.937,0.937,0.937}}751 & {\cellcolor[rgb]{0.753,0.753,0.753}}1.41 & \textbf{1.09} & {\cellcolor[rgb]{0.937,0.937,0.937}}798 & {\cellcolor[rgb]{0.753,0.753,0.753}}2.23 & 2.24 & {\cellcolor[rgb]{0.937,0.937,0.937}}109 & {\cellcolor[rgb]{0.753,0.753,0.753}}3.65 & 1.26 & {\cellcolor[rgb]{0.937,0.937,0.937}}1594 & {\cellcolor[rgb]{0.753,0.753,0.753}}29.71 & 1.66 & {\cellcolor[rgb]{0.937,0.937,0.937}}79 & {\cellcolor[rgb]{0.753,0.753,0.753}}29.71 & 7.97 & {\cellcolor[rgb]{0.937,0.937,0.937}}4126 & {\cellcolor[rgb]{0.753,0.753,0.753}}28.42 \\
			& 2\,048 & 1.11 & {\cellcolor[rgb]{0.937,0.937,0.937}}760 & {\cellcolor[rgb]{0.753,0.753,0.753}}1.34 & \textbf{1.10} & {\cellcolor[rgb]{0.937,0.937,0.937}}798 & {\cellcolor[rgb]{0.753,0.753,0.753}}2.20 & 2.26 & {\cellcolor[rgb]{0.937,0.937,0.937}}108 & {\cellcolor[rgb]{0.753,0.753,0.753}}3.65 & 1.27 & {\cellcolor[rgb]{0.937,0.937,0.937}}3187 & {\cellcolor[rgb]{0.753,0.753,0.753}}57.26 & 1.69 & {\cellcolor[rgb]{0.937,0.937,0.937}}108 & {\cellcolor[rgb]{0.753,0.753,0.753}}57.26 & 8.54 & {\cellcolor[rgb]{0.937,0.937,0.937}}8123 & {\cellcolor[rgb]{0.753,0.753,0.753}}55.97 \\
			& 4\,096 & 1.13 & {\cellcolor[rgb]{0.937,0.937,0.937}}763 & {\cellcolor[rgb]{0.753,0.753,0.753}}1.46 & \textbf{1.11} & {\cellcolor[rgb]{0.937,0.937,0.937}}798 & {\cellcolor[rgb]{0.753,0.753,0.753}}2.27 & 2.27 & {\cellcolor[rgb]{0.937,0.937,0.937}}107 & {\cellcolor[rgb]{0.753,0.753,0.753}}3.65 & \multicolumn{1}{c}{-} & \multicolumn{1}{c}{{\cellcolor[rgb]{0.937,0.937,0.937}}-} & \multicolumn{1}{c}{{\cellcolor[rgb]{0.753,0.753,0.753}}-} & \multicolumn{1}{c}{-} & \multicolumn{1}{c}{{\cellcolor[rgb]{0.937,0.937,0.937}}-} & \multicolumn{1}{c}{{\cellcolor[rgb]{0.753,0.753,0.753}}-} & \multicolumn{1}{c}{-} & \multicolumn{1}{c}{{\cellcolor[rgb]{0.937,0.937,0.937}}-} & \multicolumn{1}{c}{{\cellcolor[rgb]{0.753,0.753,0.753}}-} \\
			& 8\,192 & 1.15 & {\cellcolor[rgb]{0.937,0.937,0.937}}776 & {\cellcolor[rgb]{0.753,0.753,0.753}}1.33 & \textbf{1.12} & {\cellcolor[rgb]{0.937,0.937,0.937}}810 & {\cellcolor[rgb]{0.753,0.753,0.753}}2.16 & 2.28 & {\cellcolor[rgb]{0.937,0.937,0.937}}107 & {\cellcolor[rgb]{0.753,0.753,0.753}}3.66 & \multicolumn{1}{c}{-} & \multicolumn{1}{c}{{\cellcolor[rgb]{0.937,0.937,0.937}}-} & \multicolumn{1}{c}{{\cellcolor[rgb]{0.753,0.753,0.753}}-} & \multicolumn{1}{c}{-} & \multicolumn{1}{c}{{\cellcolor[rgb]{0.937,0.937,0.937}}-} & \multicolumn{1}{c}{{\cellcolor[rgb]{0.753,0.753,0.753}}-} & \multicolumn{1}{c}{-} & \multicolumn{1}{c}{{\cellcolor[rgb]{0.937,0.937,0.937}}-} & \multicolumn{1}{c}{{\cellcolor[rgb]{0.753,0.753,0.753}}-} \\
			& 16\,384 & 1.17 & {\cellcolor[rgb]{0.937,0.937,0.937}}797 & {\cellcolor[rgb]{0.753,0.753,0.753}}1.46 & \textbf{1.14} & {\cellcolor[rgb]{0.937,0.937,0.937}}818 & {\cellcolor[rgb]{0.753,0.753,0.753}}2.16 & 2.30 & {\cellcolor[rgb]{0.937,0.937,0.937}}109 & {\cellcolor[rgb]{0.753,0.753,0.753}}3.66 & \multicolumn{1}{c}{-} & \multicolumn{1}{c}{{\cellcolor[rgb]{0.937,0.937,0.937}}-} & \multicolumn{1}{c}{{\cellcolor[rgb]{0.753,0.753,0.753}}-} & \multicolumn{1}{c}{-} & \multicolumn{1}{c}{{\cellcolor[rgb]{0.937,0.937,0.937}}-} & \multicolumn{1}{c}{{\cellcolor[rgb]{0.753,0.753,0.753}}-} & \multicolumn{1}{c}{-} & \multicolumn{1}{c}{{\cellcolor[rgb]{0.937,0.937,0.937}}-} & \multicolumn{1}{c}{{\cellcolor[rgb]{0.753,0.753,0.753}}-} \\
			\multicolumn{20}{l}{} \\
			\multicolumn{1}{c}{\multirow{5}{*}{\rotatebox{90}{\textbf{gsh-2015}}}} & 4 & {1.26} & {\cellcolor[rgb]{0.937,0.937,0.937}}22797 & {\cellcolor[rgb]{0.753,0.753,0.753}}45.57 & {\textbf{1.25}} & {\cellcolor[rgb]{0.937,0.937,0.937}}25128 & {{\cellcolor[rgb]{0.753,0.753,0.753}}56.42} & {1.60} & {{\cellcolor[rgb]{0.937,0.937,0.937}}3178} & {{\cellcolor[rgb]{0.753,0.753,0.753}}100.49} & {1.30} & {{\cellcolor[rgb]{0.937,0.937,0.937}}2151} & {{\cellcolor[rgb]{0.753,0.753,0.753}}52.83} & {1.70} & {{\cellcolor[rgb]{0.937,0.937,0.937}}2149} & {{\cellcolor[rgb]{0.753,0.753,0.753}}52.83} & {3.59} & {{\cellcolor[rgb]{0.937,0.937,0.937}}1881} & {{\cellcolor[rgb]{0.753,0.753,0.753}}22.10} \\
			& 32 & {1.33} & {\cellcolor[rgb]{0.937,0.937,0.937}}22916 & {\cellcolor[rgb]{0.753,0.753,0.753}}45.57 & {\textbf{1.32}} & {{\cellcolor[rgb]{0.937,0.937,0.937}}25522} & {{\cellcolor[rgb]{0.753,0.753,0.753}}59.00} & {2.58} & {{\cellcolor[rgb]{0.937,0.937,0.937}}3468} & {{\cellcolor[rgb]{0.753,0.753,0.753}}100.35} & {1.51} & {{\cellcolor[rgb]{0.937,0.937,0.937}}4194} & {{\cellcolor[rgb]{0.753,0.753,0.753}}52.83} & {2.74} & {{\cellcolor[rgb]{0.937,0.937,0.937}}2195} & {{\cellcolor[rgb]{0.753,0.753,0.753}}52.83} & {16.94} & {{\cellcolor[rgb]{0.937,0.937,0.937}}10216} & {{\cellcolor[rgb]{0.753,0.753,0.753}}22.10} \\
			& 128 & {1.39} & {\cellcolor[rgb]{0.937,0.937,0.937}}23275 & {\cellcolor[rgb]{0.753,0.753,0.753}}46.24 & {\textbf{1.39}} & {{\cellcolor[rgb]{0.937,0.937,0.937}}25421} & {{\cellcolor[rgb]{0.753,0.753,0.753}}56.82} & {3.09} & {{\cellcolor[rgb]{0.937,0.937,0.937}}2268} & {{\cellcolor[rgb]{0.753,0.753,0.753}}100.49} & {1.65} & {{\cellcolor[rgb]{0.937,0.937,0.937}}11623} & {{\cellcolor[rgb]{0.753,0.753,0.753}}60.19} & {3.52} & {{\cellcolor[rgb]{0.937,0.937,0.937}}2393} & {{\cellcolor[rgb]{0.753,0.753,0.753}}60.19} & {26.99} & {{\cellcolor[rgb]{0.937,0.937,0.937}}38232} & {{\cellcolor[rgb]{0.753,0.753,0.753}}36.83} \\
			& 256 & {1.42} & {\cellcolor[rgb]{0.937,0.937,0.937}}23194 & {\cellcolor[rgb]{0.753,0.753,0.753}}45.57 & {\textbf{1.41}} & {{\cellcolor[rgb]{0.937,0.937,0.937}}25403} & {{\cellcolor[rgb]{0.753,0.753,0.753}}59.94} & {3.28} & {{\cellcolor[rgb]{0.937,0.937,0.937}}2253} & {{\cellcolor[rgb]{0.753,0.753,0.753}}100.49} & {1.70} & {{\cellcolor[rgb]{0.937,0.937,0.937}}21482} & {{\cellcolor[rgb]{0.753,0.753,0.753}}77.33} & {3.87} & {{\cellcolor[rgb]{0.937,0.937,0.937}}2545} & {{\cellcolor[rgb]{0.753,0.753,0.753}}77.33} & {31.09} & {{\cellcolor[rgb]{0.937,0.937,0.937}}64247} & {{\cellcolor[rgb]{0.753,0.753,0.753}}66.29} \\
			& 512 & {1.46} & {\cellcolor[rgb]{0.937,0.937,0.937}}23299 & {\cellcolor[rgb]{0.753,0.753,0.753}}45.57 & {\textbf{1.43}} & {{\cellcolor[rgb]{0.937,0.937,0.937}}25495} & {\cellcolor[rgb]{0.753,0.753,0.753}}59.32 & {3.44} & {{\cellcolor[rgb]{0.937,0.937,0.937}}2264} & {{\cellcolor[rgb]{0.753,0.753,0.753}}100.49} & {1.77} & {{\cellcolor[rgb]{0.937,0.937,0.937}}41227} & {{\cellcolor[rgb]{0.753,0.753,0.753}}136.25} & {4.21} & {{\cellcolor[rgb]{0.937,0.937,0.937}}2751} & {{\cellcolor[rgb]{0.753,0.753,0.753}}136.25} & \multicolumn{1}{c}{-} & \multicolumn{1}{c}{{\cellcolor[rgb]{0.937,0.937,0.937}}-} & \multicolumn{1}{c}{{\cellcolor[rgb]{0.753,0.753,0.753}}-}
		\end{tabular}
	}
\end{table}

\pagebreak
\section{Conclusion}
\label{sec:conclusion}
In this work, we propose \texttt{HeiStreamE}, a buffered streaming edge partitioner that achieves state-of-the-art solution quality, and \texttt{FreightE}, a highly efficient streaming edge partitioner that uses streaming hypergraph partitioning to assign blocks to edges on the fly. \texttt{HeiStreamE} processes the input graph in batches, constructs a novel graph transformation on the per-batch graph model, extends it with global partitioning information, and partitions it with a multilevel scheme. 
Aside from hashing-based streaming partitioners, which have poor solution quality, \texttt{HeiStreamE} and \texttt{FreightE} are the only known streaming edge partitioners whose runtime and memory consumption are both linear and asymptotically independent of the number of blocks of partition $k$. Our experiments demonstrate that \texttt{HeiStreamE} and \texttt{FreightE} consistently outperform all existing (re)streaming edge partitioners with regard to vertex replication metrics and runtime respectively. Additionally, \texttt{HeiStreamE} uses less memory than high-quality (re)streaming edge partitioners at $k \geq 256$ or when the graph has far more edges than vertices, as is the case for most real-world networks. Our findings underscore \texttt{HeiStreamE} as a highly memory-efficient and effective solution for streaming edge partitioning of large-scale graphs. 

\bibliography{phdthesiscs}

\begin{appendix}
\vfill \pagebreak
\section{Instance Properties}
\begin{table}[h!]
	\centering
	\caption{Graphs used for experiments. The roadNet graphs, wiki graphs, web-Google, web-NotreDame, and all social, co-purchasing, and autonomous systems graphs were obtained from the publicly available SNAP dataset~\cite{snap}. We also used graphs from the 10th DIMACS Implementation Challenge, namely eu-2005, in-2004 and uk-2007-05~\cite{benchmarksfornetworksanalysis}. Any remaining graphs are available on the network repository website~\cite{nr-aaai15} or on the Laboratory for Web Algorithmics website~\cite{BMSB,BRSLLP,BoVWFI}. For our experiments, we converted these graphs to a vertex-stream format (METIS) while removing parallel edges, self-loops, and directions, and assigning unitary weight to all vertices and edges. If a graph was unavailable in an equivalent edge list format, we converted the METIS format into the edge list format while preserving edge order for fairness during streaming.}
	\includegraphics[width=1.0\textwidth]{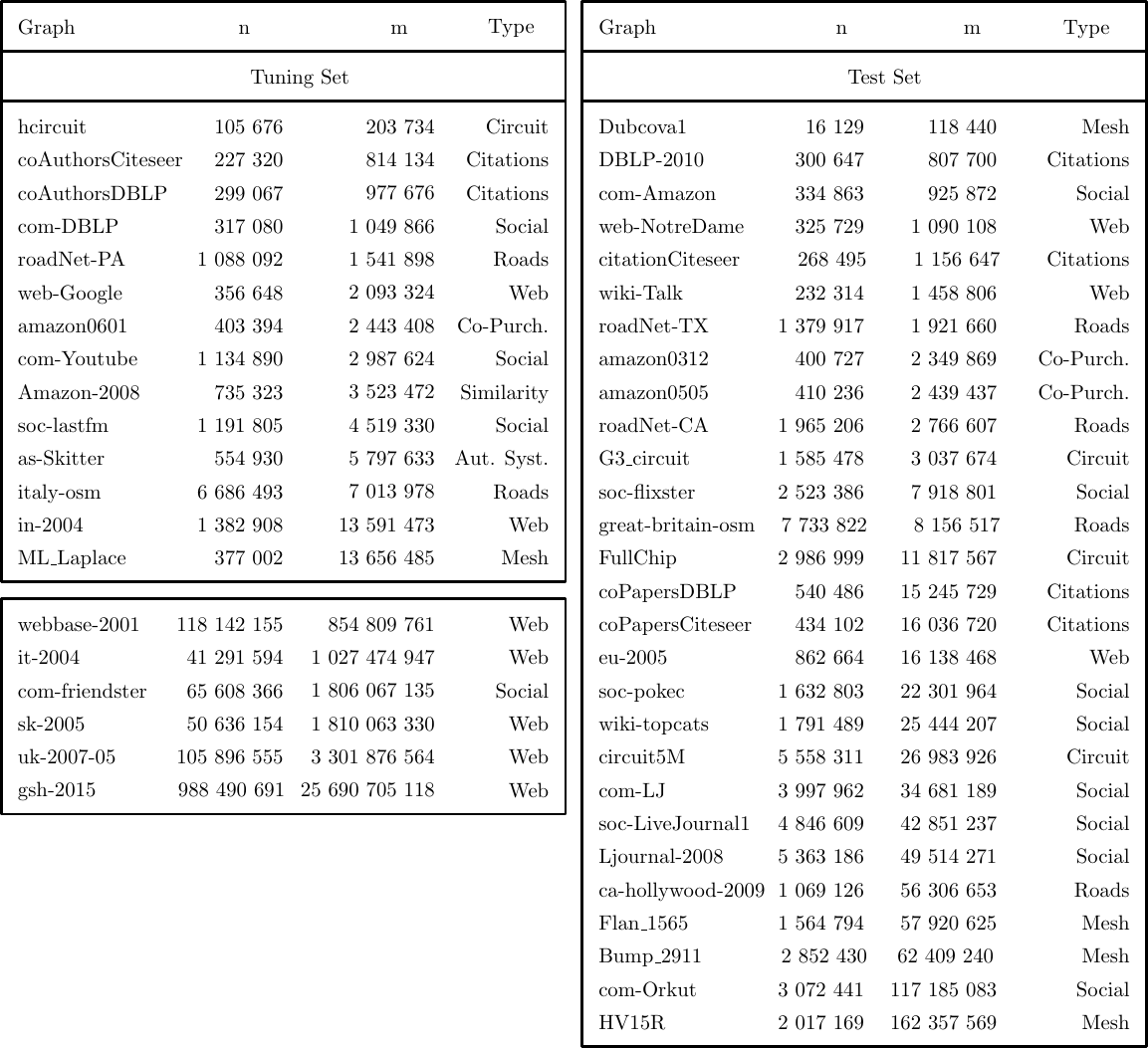}
	\label{table:graph}
\end{table}
\vfill

\end{appendix}
\end{document}